\documentclass[12pt,reqno,sort]{elsarticle}

\usepackage{amsfonts,color,amsmath,amssymb,fancyhdr}
\usepackage{amsthm}
\usepackage{graphicx}
\usepackage{mathrsfs}
\usepackage{subfigure}
\usepackage{booktabs}
\usepackage{multirow}
\usepackage{mathrsfs}
\usepackage{booktabs}
\usepackage{longtable}
\usepackage[cal=cm]{mathalfa}
\def\Box{\vcenter{\vbox{\hrule\hbox{\vrule
				\vbox to 8.8pt{\hbox to 10pt{}\vfill}\vrule}\hrule}}}

\newtheorem{thm}{Theorem}[section]
\newtheorem{lemma}[thm]{Lemma}
\newtheorem{corollary}[thm]{Corollary}
\newtheorem{prop}[thm]{Proposition}

\newtheorem{example}[thm]{Example}
\numberwithin{equation}{section}
\newtheorem{remark}[thm]{Remark}

\definecolor{Purple}{rgb}{0.5,0,0.5}

\def\a{\alpha} \def\b{\beta}

\begin{document}
	\newcommand{\stopthm}{\begin{flushright}
			\(\box \;\;\;\;\;\;\;\;\;\; \)
	\end{flushright}}
	\newcommand{\symfont}{\fam \mathfam}
	
	\title{On the equivalence of NMDS codes}
	
	\date{}

	\author[add1]{Jianbing Lu\corref{cor1}}\ead{jianbinglu@nudt.edu.cn}\cortext[cor1]{Corresponding author}
	\author[add1]{Yue Zhou}\ead{yue.zhou.ovgu@gmail.com}
	\address[add1]{Department of Mathematics, National University of Defense Technology, Changsha 410073, China}
	
	\begin{abstract}
		An $[n,k,d]$ linear code is said to be maximum distance separable (MDS) or almost maximum distance separable (AMDS) if $d=n-k+1$ or $d=n-k$, respectively. If a code and its dual code are both AMDS, then the code is  said to be near maximum distance separable (NMDS). For $k=3$ and $k=4$, there are many constructions of NMDS codes by adding some suitable projective points to arcs in $\mathrm{PG}(k-1,q)$. In this paper, we consider the monomial equivalence problem for some NMDS codes with the same  weight distributions  and present  new constructions of NMDS codes.
		\newline
		
		\noindent\text{Keywords:} NMDS code; Arc; Hyperoval; Monomial equivalence
		
		\noindent\text{Mathematics Subject Classification (2020)}: 94B05 51E21 
	\end{abstract}	
	
	\maketitle
	
	\section{Introduction}\label{introduction}
	
	An  $[n,k]_q$ \emph{linear code} $\mathcal{C}$ is a $k$-dimensional subspace of $\mathbb{F}^{n}_{q}$. The  parameter $n$ is called the  \emph{length} of  $\mathcal{C}$. The vectors in  $\mathcal{C}$ are called \emph{codewords}. Let $c\in\mathcal{C}$ be a codeword.   The \emph{support} of  $c=(c_1,c_2,\dots, c_n)$ is defined  as $supp(c)=\{i\in[n]:c_i\neq0\}$, where $[n]$  denotes the set $\{1,2,\dots,n\}$.  The \emph{weight} of $c$ is the cardinality of its support, i.e.,
	$wt(c)=|supp(c)|$. An $[n,k,d]_q$ linear code $\mathcal{C}$  is  an  $[n,k]_q$ linear code with \emph{minimum distance} $d$, defined as $d=\text{min}_{0\neq c\in\mathcal{C}} wt(c)$. Usually, if the context is clear, the subscript $q$ is omitted  by convention in the sequel. A \emph{generator matrix} of $\mathcal{C}$ is a $k\times n$ matrix $G$ whose rows form a basis of $\mathcal{C}$ as a $\mathbb{F}_q$-vector space. The \emph{dual code} of $\mathcal{C}$ is  its orthogonal subspace $\mathcal{C}^{\perp}$ with respect to the Euclidean inner product, defined as $\mathcal{C}^{\perp}=\{x\in\mathbb{F}^{n}_{q}: x\cdot c=0, \forall c\in\mathcal{C}\}$.  A \emph{monomial matrix} over $\mathbb{F}_q$ is a square matrix with exactly one nonzero element from $\mathbb{F}_q$ in each row and column. Such a  matrix $M$ can be expressed as $DP$ or  $PD_1$, where $D$ and $D_1$ are diagonal matrices and $P$ is a permutation matrix. Two $[n,k,d]$ linear codes $C_1$ and $C_2$, with generator matrices $G_1$ and $G_2$ respectively, are \emph{monomially equivalent} if  there exist  a matrix $L\in \mathrm{GL}(k,q)$ and an $n\times n$ monomial matrix $M$ such that $LG_1M=G_2$. 
	
	The Singleton bound \cite{Singleton1964} states a relationship among $n,k$ and $d$: $d\leq n-k+1$. The codes meeting this bound are called \emph{maximum distance separable} (MDS) codes. The famous MDS conjecture states that for a nontrivial MDS code, the maximum length $n$ satisfies:  $n\leq q+2$ if $q$ is even and $k=3$ or $k=q-1$, and  $n\leq  q+1$ otherwise \cite{Segre1955}.  This conjecture was extensively investigated  by Ball in \cite{Ball2012,Ball}. 	A code  with parameters of the form $[n,k,n-k]$ is said to be \emph{almost maximum distance separable} (AMDS). If both a linear code and its dual code are AMDS,   the   code is 
	said to be \emph{near maximum distance separable} (NMDS), which  was introduced by Dodunekov and Landjev \cite{Dodunekov1995}. Let  $m(k,q)$ denote the maximum possible length  for a  nontrivial $q$-ary NMDS code. Similarly to the MDS conjecture, the following conjecture for NMDS codes was proposed in \cite{Dodunekov1995,Landjev2015}: $m(k,q)\leq 2q+2$. Specially, for $k=3$ and $q\geq8$, it is conjecdtured that $m(3,q)\leq2q-1$, and it is known $m(3,q)=2q-1$ for $q=8,9$.
	
	Both MDS codes and NMDS codes can be investigated in the language of finite geometries. Let $\mathrm{PG}(r,q)$ denote the  $r$-dimensional projective space derived from $\mathbb{F}^{r+1}_q$, where its points correspond to the one-dimensional subspaces of $\mathbb{F}^{r+1}_q$, and its hyperplanes  correspond to the $r$-dimensional subspaces of $\mathbb{F}^{r+1}_q$. An \emph{$n$-arc} in $\mathrm{PG}(k-1,q)$ is a set of $n$ points such that  every hyperplane intersects  the arc in at most $k-1$ points. An arc is said to be \emph{maximal} if it has the largest possible number of points among all arcs. For an  $[n,k,d]$ linear code $\mathcal{C}$ with  generator matrix $G$, let $\mathcal{S}_{G}$ be the set of points in $\mathrm{PG}(k-1,q)$ corresponding to the  columns of $G$. For any nonzero vector $u\in\mathbb{F}^{k}_q$, we can define the  hyperplane $\mathcal{H}_u=\{x\in\mathbb{F}^{k}_q: x\cdot u=0\}$ in $\mathrm{PG}(k-1,q)$. The following lemma establishes a relationship between the weights of codewords in $\mathcal{C}$ and the hyperplanes of $\mathrm{PG}(k-1,q)$.		
	\begin{lemma}\label{relation}
		\cite[Lemma 5.16]{Ball2015} Let $u$ be a nonzero vector of $\mathbb{F}^{k}_q$. The codeword $uG$ has weight $\omega$ if and only if the hyperplane $\mathcal{H}_u$ in $\mathrm{PG}(k-1,q)$ contains $ |\mathcal{S}_{G}|-\omega$ points of $\mathcal{S}_{G}$, i.e., $|\mathcal{H}_u\cap \mathcal{S}_{G}|=n-wt(uG)$. Especially, any hyperplane in $\mathrm{PG}(k-1,q)$ intersects $\mathcal{S}_{G}$ in at most $n-d$ points. 
	\end{lemma}
	From Lemma \ref{relation}, we see that $[n,k,n-k+1]$ MDS codes and $n$-arcs in $\mathrm{PG}(k-1,q)$  are indeed equivalent objects. Similarly, an $[n,k,n-k]$ NMDS code can be viewed as a point-set $\mathcal{S}$ of size $n$ in $\mathrm{PG}(k-1,q)$  satisfying the following conditions \cite{Dodunekov1995}: 
	\begin{enumerate}
		\item[(N1)] any $k-1$ points of $\mathcal{S}$ generate a hyperplane in $\mathrm{PG}(k-1,q)$;
		\item[(N2)] there exist $k$ points in $\mathcal{S}$  that lie on a hyperplane in $\mathrm{PG}(k-1,q)$;
		\item[(N3)] every $k+1$ points of $\mathcal{S}$ generate $\mathrm{PG}(k-1,q)$.
	\end{enumerate}
	By adding suitable projective points to the $n$-arcs in $\mathrm{PG}(k-1,q)$, many  NMDS codes of dimension $k$ have been constructed. For the case $k=3$, see \cite{Wang2021,Li20,Li2023,Xu2024,Fan2024}. For the cases $k=4$ and $k=5$, see \cite{Xu2023,Heng,Ding2024,Ceria}. Moreover, there are many NMDS codes with the same parameters derived from these constructions. In this paper, we investigate the monomial equivalence problem  for such NMDS codes  and provide some new constructions. Specifically, by  utilizing hyperovals in $\mathrm{PG}(2,q)$, we propose an infinite family of NMDS codes with length $q+3$ and dimension $3$ for even $q$ (see Theorem \ref{main}). Through detailed analysis of the action of  the homography stabilizers of distinct hyperovals   on  the points not in hyperovals, we fully determine the equivalence classes of these NMDS codes. By analyzing the orbits of points and planes  of $\mathrm{PG}(3,q)$ under the  stabilizers of arcs,  we provide a unified geometric perspective for constructing several equivalence classes of NMDS codes with length $q+2$ and dimension $4$, derived from arcs in $\mathrm{PG}(3,q)$  with additional points (see Theorems \ref{4_even} and \ref{4_odd} ). 
	
	This paper is organized as follows. In Section \ref{s2}, we present some preliminary results on NMDS codes and the maximal arcs in $\mathrm{PG}(2,q)$ and $\mathrm{PG}(3,q)$.  The monomial equivalence problem for  NMDS codes of dimensions $3$ and $4$ is considered  in Sections \ref{s3} and \ref{s4}, respectively. In Section \ref{s5}, we conclude this paper.
	
	\section{Preliminaries}\label{s2}
	In this section, we present  some preliminary results which will be used later.
	\subsection{NMDS codes}
	Let $\mathcal{C}$ be an  $[n,k,n-k]$ NMDS code with generator matrix $G=(\mathbf{g}_1,\mathbf{g}_2,\ldots,\mathbf{g}_{n})$, where $\mathbf{g}_i\in \mathbb{F}^k_q$ and $k\geq3$. Since $k\geq3$,  the columns of $G$   can be  identified with distinct points in $\mathrm{PG}(k-1,q)$. Define  $\mathcal{S}_G=\{\mathbf{g}_1,\mathbf{g}_2,\ldots,\mathbf{g}_{n}\}$. Then $\mathcal{S}_G$ is a point-set in $\mathrm{PG}(k-1,q)$ satisfying the conditions (N1), (N2) and (N3) stated in Section \ref{introduction}.
	
	Let $A_{i}$ be the number of codewords with weight $i$ in $\mathcal{C}$, where $0\leq i\leq n$. The polynomial $A(z)=1+A_1z+A_2z^2+\cdots+A_nz^n$ is called the \emph{weight enumerator}  of $\mathcal{C}$. The sequence $(1,A_1,A_2,\dots,A_n)$ is termed the \emph{weight distribution} of $\mathcal{C}$. The  following lemma gives the weight distribution formulas for NMDS codes.
	\begin{lemma}\label{enumerators}
		\cite{Dodunekov1995}
		Let $\mathcal{C}$ be an $[n, k, n-k]$ NMDS code. Then the weight distribution of $\mathcal{C}$ is given by
		$$
		A_{n-k+s}=\binom{n}{k-s} \sum_{j=0}^{s-1}(-1)^j\binom{n-k+s}{j}\left(q^{s-j}-1\right)+(-1)^s\binom{k}{s} A_{n-k}
		$$	for $s \in\{1,2, \ldots, k\}$.
	\end{lemma}
	It follows from Lemma \ref{enumerators} that the weight distribution  of an NMDS code can be determined by the number of its minimum weight codewords, i.e., $A_{n-k}$. Let $\mathcal{H}_u$ be the hyperplane in $\mathrm{PG}(k-1,q)$ determined by the nonzero vector $u\in\mathbb{F}^k_q$. Then by Lemma \ref{relation}, we have 	\begin{equation}\label{A_{n-k}}
		A_{n-k}=\#\left\{ u \in \mathbb{F}_q^k\setminus\{\mathbf{0}\}: |\mathcal{H}_{u} \cap \mathcal{S}_G|=k\right\}.
	\end{equation}
	Here the notation $\#A$ denotes the cardinality of set $A$.
	\subsection{Hyperovals}
	An $n$-arc in $\mathrm{PG}(2,q)$ is a set of $n$ points such that  no three  are collinear. It is known that $n\leq q+2$,  with equality holding if and only if $q$ is even. Let $q=2^m$ for a positive integer $m$.  A \emph{hyperoval} $\mathcal{HO}$ in $\mathrm{PG}(2,q)$ is  an arc with $q+2$ points, where every line in $\mathrm{PG}(2,q)$ intersects  $\mathcal{HO}$ in either $0$ or $2$ points. Two hyperovals are equivalent if one can be mapped to the other via  an automorphism of $\mathrm{PG}(2,q)$. 
	All hyperovals in $\mathrm{PG}(2,q)$ can be constructed using a special type of permutation polynomials over $\mathbb{F}_q$, as described in the following Theorem \cite{Hirschfeld1979}.
	
	\begin{thm}\label{hyperoval}
		Let $q>2$ be a power of $2$. Any hyperoval in $\mathrm{PG}(2,q)$ can be written in the following form
		\[
		\mathcal{HO}(f)=\{(1,c,f(c)):c\in\mathbb{F}_q\}\cup\{(0,0,1)\}\cup\{(0,1,0)\},
		\]
		where $f(x)\in\mathbb{F}_{q}[x]$ is a permutation polynomial of $\mathbb{F}_{q}$ satisfying:
		\begin{enumerate}
			\item[(1)] $\text{deg}(f)\leq q-2$ and $f(0)=0$, $f(1)=1$;
			\item[(2)] for each $a\in\mathbb{F}_{q}$, $f_a(x)=(f(x+a)+f(a))/x$ is also a permutation polynomial of $\mathbb{F}_{q}$, and $f_a(0)=0$. 
		\end{enumerate}
		Conversely, every   set $\mathcal{HO}(f)$ defined in this way is a hyperoval.
	\end{thm}
	
	By \cite{Bose1947}, for  even $q$, a conic together with its nucleus (the intersection point of its tangents) forms a hyperoval. Such a hyperoval is said to be \emph{regular}. As shown by \cite{Segre1957}, for $q=2,4,8$, every hyperoval is regular. However, for $q=2^m$ with $m\geq4$, irregular hyperovals exist. 	Any polynomial satisfying the two conditions in Theorem \ref{hyperoval} is called an \emph{oval polynomial} (or \emph{o-polynomial} for short).
	Table \ref{hyperovals} gives a list of all known  o-polynomials over $\mathbb{F}_q$ with $q=2^m$, also see \cite[Table 3]{Hirschfeld2001}. Here,
	$$
	F_1(x)=\left(\delta^2\left(x^4+x\right)+\delta^2\left(1+\delta+\delta^2\right)\left(x^3+x^2\right)\right)\left(x^4+\delta^2 x^2+1\right)^{-1}+x^{1/2};
	$$
	\[F_3(x)=x^4+x^{16}+x^{28}+\omega^{11}(x^6+x^{10}+x^{14}+x^{18}+x^{22}+x^{26})+\omega^{20}(x^8+x^{20})+\omega^6(x^{12}+x^{24});\]
	$$
	F_2(x)=\frac{T\left(\beta^m\right)(x+1)}{T(\beta)}+\frac{T\left(\left(\beta x+\beta^q\right)^m\right)}{T(\beta)\left(x+T(\beta) x^{2^{m-1}}+1\right)^{m-1}}+x^{2^{m-1}}
	$$ and condition$(*)$ is that $m \geq 4$ is even, $\beta \in \mathbb{F}_{q^2} \backslash\{1\}$ with $\beta^{q+1}=1, m \equiv \pm(q-1) / 3 \pmod{q+1}$, and $T(x)=x+x^q$.

	\begin{table}
		\begin{center}
			\caption{Hyperovals in $\mathrm{PG}(2,q)$, $q$ even}\label{hyperovals}
			\resizebox{\textwidth}{!}{
				\begin{tabular}{ccccc}
					\hline Name & $f(x)$ & Conditions  & Reference\\
					\hline
					Regular & $x^{2}$ & & \cite{Bose1947} \\
					Irregular translation & $x^{2^h}$ & $\gcd(h, m)=1$, $h>1$ & \cite{Segre1957} \\
					Segre & $x^{6}$ & $m$  odd & \cite{Segre1962} \\ 
					Glynn I & $x^{3 \times 2^{(m+1) / 2}+4}$ & $m$  odd & \cite{Glynn1983} \\ 
					Glynn II & $x^{2^{(m+1) / 2}+2^{(m+1) / 4}}$ & $m \equiv 3 \pmod{4}$ & \cite{Glynn1983} \\
					& $x^{2^{(m+1) / 2}+2^{(3 m+1) / 4}}$ & $m \equiv 1 \pmod{4}$ & \cite{Glynn1983} \\
					Cherowitzo & $x^{2^e}+x^{2^e+2}+x^{3 \times 2^e+4}$ & $e=(m+1) / 2$, $m$   odd & \cite{Cherowitzo1998} \\
					Payne & $x^{1/6}+x^{3/6}+x^{5/6}$ &   $m$   odd & \cite{Payne1985} \\
					Subiaco & $F_1(x)$ &   $\operatorname{Tr}_{q / 2}(1 /\delta)=1$, $\delta \notin \mathbb{F}_4$ if $m \equiv 2 \bmod 4$ & \cite{Payne1995} \\
					Adelaide & $F_2(x)$ &   $(*)$ & \cite{Cherowitzo2003} \\	
					O'Keefe-Penttila & $F_3(x)$ &   $\omega^5=\omega^2+1$, $m=5$ & \cite{O'Keefe1992} \\
					\hline
			\end{tabular}}
		\end{center}
	\end{table}

	Some useful properties on hyperovals and o-polynomials are presented below.
	
	\begin{lemma}\label{unique}
		\cite[Lemma 4.6.6]{Payne}  If hyperovals $K$ and $K'$ in $\mathrm{PG}(2,q)$  have $\frac{1}{2}(q+2)+n$ points in common,	with $n>0$, then $K=K'$.
	\end{lemma}
	
	\begin{lemma}\label{two}
		\cite{Maschietti1998} A polynomial $f(x)$ over $\mathbb{F}_{q}$ with $f(0)=0$ is an o-polynomial if and only if  $f(x)+ux$ is 2-to-1 for every $u\in\mathbb{F}^{*}_{q}$.
	\end{lemma}
	
	\begin{lemma}\label{number}
		Let $f(x)\in\mathbb{F}_{q}[x]$ be an  o-polynomial, $u_1,u_2\in\mathbb{F}^{*}_{q}$, $a,b\in\mathbb{F}_{q}$ with $(a,b)\neq(0,0)$. Define
		\[A(u_1,u_2)=\#\{x\in\mathbb{F}_q:u_1f(x)+u_2x+u_1a+u_2b=0\}.\]
		Then $A(u_1,u_2)\in\{0,2\}$ and $\#\{(u_1,u_2)\in(\mathbb{F}_q^*)^2:A(u_1,u_2)=2\}=(q-1)(q-2)/2.$
	\end{lemma}
	\begin{proof}
		The proof is similar to  that of \cite[Lemma 3.1]{Xu2024}, and so we omit the details.
	\end{proof}
	
	\begin{lemma}\label{root}
		\cite{Wan2003}	Let $r$ be a positive integer, and $f(x)=x^r$ be a polynomial over 	$\mathbb{F}_{q}$. 
		\begin{enumerate}
			\item[(1)] If $\gcd(r,q-1)=1$, then $f(x)$ is a permutation polynomial over $\mathbb{F}_{q}$. 
			\item[(2)] If $\gcd(r,q-1)=r$, then $f(x)$ is $r$-to-$1$ over $\mathbb{F}^{*}_{q}$.
		\end{enumerate}
	\end{lemma}
	\subsection{Arcs in $\mathrm{PG}(3,q)$}
	An arc in $\mathrm{PG}(3,q)$, $q\geq4$ has at most $q+1$ points \cite[Theorems 21.2.4 and 21.3.8]{Hirschfeld1985}. If $q=2^m$, a $(q+1)$-arc in $\mathrm{PG}(3,q)$ is projectively equivalent to the point-set
	$$\mathcal{S}_h=\{P_x=(1,x,x^{2^h},x^{2^h+1})\mid x\in\mathbb{F}_q\}\cup\{P_{\infty}=(0,0,0,1)\},$$
	where $\gcd(m,h)=1$ \cite[Theorem 21.3.15]{Hirschfeld1985}. For $h=1$,   $\mathcal{S}_1$ is called a \emph{twisted cubic}. When  $q$ is odd, every $(q+1)$-arc in $\mathrm{PG}(3,q)$ is a twisted cubic. A line in $\mathrm{PG}(3,q)$ joining two distinct points of $\mathcal{S}_h$ is called a \emph{real chord}.  Note that no four points of $\mathcal{S}_h$ are coplanar. Therefore, each real chord contains exactly two points of $\mathcal{S}_h$. Let $R$ denote the set of points lying on real chords but not in $\mathcal{S}_h$. Since  there are $q(q+1)/2$ real chords, it follows that $|R|=(q^3-q)/2$. The planes $\pi(P_x)$ and $\pi(P_{\infty})$, defined by the equations $x^{2^h+1}X_1+x^{2^h}X_2+xX_3+X_4=0$ and $X_1=0$ are called the \emph{osculating planes} of $\mathcal{S}_h$  at $P_x$ and $P_{\infty}$, respectively. The osculating plane $\pi(P)$ intersects   $\mathcal{S}_h$ only at point $P$. We end this section with the following lemma.
	
	\begin{lemma}\label{unique2}
		\cite[Theorem A]{Storme1993} Let $K$ be a $k$-arc of $\mathrm{PG}(3,q)$, $q$ even and $q\neq2$. Assume that $k>q-\sqrt{q}/2+9/4$. Then  $K$ lies on a unique $(q+1)$-arc.
	\end{lemma}
	\section{NMDS codes of dimension $3$}\label{s3}
	\subsection{NMDS codes from hyperovals}
	When $k=3$, a point-set $\mathcal{S}$ of size $n$ in  $\mathrm{PG}(2,q)$ satisfying the properties (N2) and (N3) in Section \ref{introduction}, i.e., $|\mathcal{H}_u\cap \mathcal{S}|\leq3$ for any nonzero vector $u\in\mathbb{F}^{3}_q$ and $|\mathcal{H}_v\cap \mathcal{S}|=3$ for some  nonzero vector $v\in\mathbb{F}^{3}_q$, is called an $(n,3)$-arc. Every $[n,3,n-3]$ NMDS code is equivalent to an $(n,3)$-arc in $\mathrm{PG}(2,q)$. In this section, let $q=2^m$ and $\mathbb{F}^{*}_{q}=\mathbb{F}_{q}\setminus\{0\}=\{\alpha_1,\dots,\alpha_{q-1}\}$, where $m$ is a positive integer.
	Define $f(x)\in\mathbb{F}_{q}[x]$ to be an  o-polynomial, and let $G$ be a $3\times(q+3)$ matrix defined as
	\begin{equation}\label{generator}
		\begin{aligned}G=(\mathbf{g}_1,\ldots,\mathbf{g}_{q+2},\mathbf{g}_{q+3})=\begin{pmatrix}1&1&\cdots&1&1&0&0&a\\\alpha_1&\alpha_2&\cdots&\alpha_{q-1}&0&0&1&b\\f(\alpha_1)&f(\alpha_2)&\cdots&f(\alpha_{q-1})&0&1&0&c\end{pmatrix},\end{aligned}
	\end{equation}
	where $\mathbf{g}_{q+3}=(a,b,c)^{\mathrm{T}}\neq\mathbf{0}$, and $\mathbf{g}_{q+3}\neq\lambda\mathbf{g}_{i}$ for some $\lambda\in\mathbb{F}^{*}_{q}$ and any $i\in\{1,\dots, q+2\}$. For convenience, we identify each vector $\mathbf{g}_i$ with the point in $\mathrm{PG}(2,q)$ spanned by $\mathbf{g}_i$. This ensures that the  points  $\mathbf{g}_i$ are distinct for $i=1,\dots,q+3$. Let $\mathcal{C}$ be the linear code over $\mathbb{F}_{q}$ with  generator matrix $G$.  For any $\mathbf{g}_{q+3}$, we first prove the following  theorem, which  generalizes  \cite[Theorem 8]{Wang2021} and \cite[Theorem 10]{Li2023}. 
	
	\begin{thm}\label{main}
		Let $q=2^m$ with $m\geq2$. Then $\mathcal{C}$ is a $[q+3,3,q]$ NMDS code over $\mathbb{F}_q$ with weight enumerator
		\[A(z)=1+\frac{(q-1)(q+2)}{2}z^q+\frac{q(q-1)(q+2)}{2}z^{q+1}+\frac{q(q-1)}{2}z^{q+2}+\frac{q(q-2)(q-1)}{2}z^{q+3}.\]
	\end{thm}
	\begin{proof}
		Let $\mathcal{S}_{G}$ denote the set of  columns of $G$. Since the first $q+2$ columns of $\mathcal{S}_{G}$ form a hyperoval $\mathcal{HO}$, the condition    $|\mathcal{H}_u\cap \mathcal{S}_{G}|\leq3$ holds for any nonzero vector $u\in\mathbb{F}^{3}_q$.  To show that  $\mathcal{S}_{G}$ is an $(n,3)$-arc  and thus $\mathcal{C}$ is an NMDS code, we calculate $A_{q}$, i.e., the number of vectors $u=(u_1,u_2,u_3)\in\mathbb{F}^{3}_q$ satisfying $|\mathcal{H}_u\cap \mathcal{S}_{G}|=3$.  Note that if $|\mathcal{H}_u\cap \mathcal{S}_{G}|=3$, then we must have $\mathbf{g}_{q+3}\in\mathcal{H}_u$. We now consider each possible case for $\mathbf{g}_{q+3}$ in turn.
		
		\textbf{Case 1}.  $\mathbf{g}_{q+3}=(a,b,c)^{\mathrm{T}}=(0,1,c)^{\mathrm{T}}$ with $c\in\mathbb{F}^{*}_q$.
		
		Then $u_2+cu_3=0$. If $u_2=0$, then $u_3=0$ and thus $u=(u_1,0,0)$. In this case, we have $\mathcal{H}_u\cap \mathcal{S}_{G}=\{\mathbf{g}_{q+1},\mathbf{g}_{q+2},\mathbf{g}_{q+3}\}$. If $u_2\neq0$, then $\mathbf{g}_{q+1}$ and $\mathbf{g}_{q+2}$ are not in $\mathcal{H}_u$. Suppose that $(1,x,f(x))^{\mathrm{T}}\in\mathcal{H}_u$ for some $x\in\mathbb{F}_q$. Then $u_1+cu_3x+u_3f(x)=0$, i.e., $f(x)+cx+u_1u_3^{-1}=0$. By Lemma \ref{two}, this equation has exactly two solutions  and  there are $q(q-1)/2$ choices for $(u_1,u_3)$ satisfying $u_1u_3^{-1}\in\{f(x)+cx:x\in\mathbb{F}_q\}$. Thus, in this case, $A_q=(q-1)+q(q-1)/2=(q-1)(q+2)/2$.
		
		\textbf{Case 2}.  $\mathbf{g}_{q+3}=(a,b,c)^{\mathrm{T}}=(1,0,c)^{\mathrm{T}}$ with $c\in\mathbb{F}^{*}_q$.
		
		Then $u_1+cu_3=0$. If $u_1=0$, then $u_3=0$ and thus $u=(0,u_2,0)$.  In this case, we have $\mathcal{H}_u\cap \mathcal{S}_{G}=\{\mathbf{g}_{q},\mathbf{g}_{q+1},\mathbf{g}_{q+3}\}$. If $u_1\neq0$, then $\mathbf{g}_{q}$ and $\mathbf{g}_{q+1}$ are not in $\mathcal{H}_u$. Moreover, if $u_2=0$, then $\mathcal{H}_u\cap \mathcal{S}_{G}=\{\mathbf{g}_{j},\mathbf{g}_{q+2},\mathbf{g}_{q+3}\}$, where $f(\alpha_j)=c$. Finally, we consider the case where $u_1,u_2\neq0$. Suppose that $(1,x,f(x))^{\mathrm{T}}\in\mathcal{H}_u$ for some $x\in\mathbb{F}^{*}_q$. Then $cu_3+u_2x+u_3f(x)=0$. By Lemma \ref{number}, this equation has exactly two solutions   and   there are $(q-1)(q-2)/2$ choices for $(u_1,u_2)$. Thus, in this case, $A_q=2(q-1)+(q-1)(q-2)/2=(q-1)(q+2)/2$.
		
		\textbf{Case 3}.  $\mathbf{g}_{q+3}=(a,b,c)^{\mathrm{T}}=(1,b,0)^{\mathrm{T}}$ with $b\in\mathbb{F}^{*}_q$.
		
		This case is similar to the Case 2 and we omit the details.
		
		\textbf{Case 4}.  $\mathbf{g}_{q+3}=(a,b,c)^{\mathrm{T}}=(1,b,c)^{\mathrm{T}}$ with $b,c\in\mathbb{F}^{*}_q$ and $c\neq f(b)$. 
		
		Then $u_1+bu_2+cu_3=0$. If $u_2=0$ and $u_3\neq0$ , then  $\mathcal{H}_u\cap \mathcal{S}_{G}=\{\mathbf{g}_{i},\mathbf{g}_{q+2},\mathbf{g}_{q+3}\}$, where $f(\alpha_i)=c$. If $u_3=0$ and $u_2\neq0$ , then  $\mathcal{H}_u\cap \mathcal{S}_{G}=\{\mathbf{g}_{j},\mathbf{g}_{q+2},\mathbf{g}_{q+3}\}$, where $\alpha_j=b$. Finally, if $u_2\neq0$ and $u_3\neq0$, and  $(1,x,f(x))^{\mathrm{T}}\in\mathcal{H}_u$ for some $x\in\mathbb{F}^{*}_q$, then
		$u_1+u_2x+u_3f(x)=0$, which can be rewritten as $u_3f(x)+u_2x+u_3c+u_2b=0$.  By Lemma \ref{number}, this equation has exactly two solutions and   there are $(q-1)(q-2)/2$ choices for $(u_1,u_2)$. Thus, in this case, $A_q=2(q-1)+(q-1)(q-2)/2=(q-1)(q+2)/2$.
		
		From Case 1-Case 4,  we see that the number of codewords of weight $q$ in $\mathcal{C}$ is $(q-1)(q+2)/2$ for any possible   $\mathbf{g}_{q+3}$. The desired weight enumerator of  $\mathcal{C}$ then follows from Lemma \ref{enumerators}. This completes the proof.
	\end{proof}
	
	A \emph{collineation} $\theta(A,\sigma)$ of $\mathrm{PG}(2,q)$ is a projective semilinear map defined by:
	$(x,y,z)^{\mathrm{T}}\mapsto A(x^{\sigma},y^{\sigma},z^{\sigma})^{\mathrm{T}}$,
	where $A$ is a nonsingular $3\times3$ matrix   over $\mathbb{F}_q$ and $\sigma\in\textup{Aut}(\mathbb{F}_q)$ (for convenience, we represent points of $\mathrm{PG}(2,q)$ as column vectors).  The full collineation group of $\mathrm{PG}(2,q)$ is $\mathrm{P}\Gamma\mathrm{L}(3,q)$. If $\sigma=1$, the collineation  is called a \emph{homography}. All homographies form the group $\mathrm{PGL}(3,q)$. The stabilizer of a hyperoval $\mathcal{HO}$ in $\mathrm{P}\Gamma\mathrm{L}(3,q)$ is called  the \emph{collineation stabilizer} of  $\mathcal{HO}$, while its stabilizer  in $\mathrm{PGL}(3,q)$ is called the \emph{homography stabilizer}.

	\begin{prop}\label{equivalent}
		Let $q=2^m$ with $m\geq2$. Suppose that $\mathcal{C}_{1}$ and  $\mathcal{C}_{2}$ are $[n,3,n-3]$ NMDS codes with generator matrix $G_1$ and $G_2$ respectively, where $q+2<n<5(q+2)/4$. If  \[G_1=(\mathbf{g}_1,\ldots,\mathbf{g}_{q+2},\mathbf{g}_{q+3},\dots,\mathbf{g}_{n}) \; \text{and} \; G_2=(\mathbf{g}'_1,\ldots,\mathbf{g}'_{q+2},\mathbf{g}'_{q+3},\dots,\mathbf{g}'_{n}),\] where  the point-sets $\{\mathbf{g}_1,\ldots,\mathbf{g}_{q+2}\}$  and $\{\mathbf{g}'_1,\ldots,\mathbf{g}'_{q+2}\}$ are   hyperovals $\mathcal{HO}_1$ and $\mathcal{HO}_2$   in $\mathrm{PG}(2,q)$ respectively, then $\mathcal{C}_{1}$ and  $\mathcal{C}_{2}$ are monomially equivalent if and only if  there exists a homography  $\theta$ such that $\mathcal{HO}^{\theta}_1=\mathcal{HO}_2$ and the sets of points $\{\mathbf{g}_{q+3}^{\theta},\dots,\mathbf{g}_{n}^{\theta}\}$ and $\{\mathbf{g}'_{q+3},\dots,\mathbf{g}'_{n}\}$ are equal. 
	\end{prop}
	\begin{proof}
		Recall that $\mathcal{C}_{1}$ and  $\mathcal{C}_{2}$  are monomially equivalent if and only if there exist  a matrix $L\in \mathrm{GL}(k,q)$ and an $n\times n$ monomial matrix $M$ such that $LG_1M=G_2$. This is equivalent to  the sets of projective points  \[\{L\mathbf{g}_1,\ldots,L\mathbf{g}_{q+2},L\mathbf{g}_{q+3},\dots,L\mathbf{g}_{n}\} \; \text{and} \; \{\mathbf{g}'_1,\ldots,\mathbf{g}'_{q+2},\mathbf{g}'_{q+3},\dots,\mathbf{g}'_{n}\} \] are equal. We can write $L\mathbf{g}_i=\mathbf{g}_{i}^{\theta}$ for some homography $\theta$, where $i=1,\dots,q+2$. Since  $\{\mathbf{g}_1,\ldots,\mathbf{g}_{q+2}\}$ is a hyperoval $\mathcal{HO}$ in $\mathrm{PG}(2,q)$, the set $\mathcal{HO}^{\theta}_1:= \{\mathbf{g}^{\theta}_1,\ldots, \mathbf{g}^{\theta}_{q+2}\}$ is also a hyperoval. Given $n<5(q+2)/4$, it follows that \[|\mathcal{HO}_2\cap\mathcal{HO}^{\theta}_1|\geq q+2-2(n-q-2)>(q+2)/2.\] By Lemma \ref{unique}, we conclude that $\mathcal{HO}_2=\mathcal{HO}^{\theta}_1$, and thus $\{\mathbf{g}_{q+3}^{\theta},\dots,\mathbf{g}_{n}^{\theta}\}=\{\mathbf{g}'_{q+3},\dots,\mathbf{g}'_{n}\}$. This completes the proof.
	\end{proof}
	\begin{corollary}\label{equivalent1}
		Let $\mathscr{C}$ be the set of $[n,3,n-3]$ NMDS  codes  over $\mathbb{F}_{q}$ with $q+2<n<5(q+2)/4$, such that the first $q+2$ columns of their generator matrices form a hyperoval $\mathcal{HO}$ in $\mathrm{PG}(2,q)$. Then   the number of  monomial equivalence classes of NMDS codes  in  $\mathscr{C}$ is   equal to the number of orbits of the homography stabilizer of $\mathcal{HO}$  acting on the set formed by the remaining 
		$n-q-2$ columns of the generator matrices.
	\end{corollary}
	\subsection{Monomial equivalence classes of NMDS codes}
	
	Recall that a $[q+3,3,q]$ NMDS code was constructed in Theorem \ref{main}. In this subsection, for the  cases  where $f(x)$ in \eqref{generator} corresponds to different o-polynomials, we determine the number of   monomial equivalence classes for the corresponding  NMDS codes generated by \eqref{generator}. Let $\mathcal{HO}$ denote  the hyperoval corresponding to $\{\mathbf{g}_1,\ldots,\mathbf{g}_{q+2}\}$  in \eqref{generator} and $M$ be the homography stabilizer of $\mathcal{HO}$. 
	\begin{lemma}\label{regular}
		\cite[Corollary 6]{Hirschfeld1979}\cite[Theorem 4.11.4]{Payne} A regular hyperoval in $\mathrm{PG}(2,q)$, $q>4$ even, has homography stabilizer of order $q(q^2-1)$ isomorphic to $\mathrm{PGL}(2,q)$. This stabilizer consists of the maps:  $(x,y,z)^{\mathrm{T}}\mapsto A(x, y,z )^{\mathrm{T}}$, where 
		$$\left.A=\left(\begin{array}{ccc}a&0&c\\\sqrt{ab}&1&\sqrt{cd}\\b&0&d\end{array}\right.\right),\mathrm{~and~}ad+bc=1.$$
	\end{lemma}
	\begin{thm}\label{re}
		Let $\mathscr{C}$ be the set of NMDS codes over $\mathbb{F}_{q}$ with  generator matrices of the form  \eqref{generator}. If $f(x)$ is the regular o-polynomial, i.e., $f(x)=x^2$, then there is only one equivalence class in $\mathscr{C}$.
	\end{thm}
	\begin{proof}
		We  provide details only for $q\geq8$, as the cases $q=2$ or $4$ can be verified using  Magma \cite{Bosma}. Note that the generator  matrices of all NMDS codes in $\mathscr{C}$ differ only in the last column vector $\mathbf{g}_{q+3}$.  To determine the number of equivalence classes in  $\mathscr{C}$,  it suffices to consider  the $q^2-1$ cases where $\mathbf{g}_{q+3}$ is different from the points in $\mathcal{HO}$. Let $\mathcal{C}$ be the NMDS code in $\mathscr{C}$ with $\mathbf{g}_{q+3}=(0,1,1)^{\mathrm{T}}$. Recall that $M$ is the homography stabilizer of $\mathcal{HO}$. We regard  $\mathbf{g}_{q+3}$ as a point in $\mathrm{PG}(2,q)$ and  compute the stabilizer $M_{\mathbf{g}_{q+3}}$ of $\mathbf{g}_{q+3}$ in $M$. By Lemma \ref{regular}, for any $\theta\in M$, we have 
		\[\mathbf{g}_{q+3}^{\theta}=\left(\begin{array}{ccc}a&0&c\\\sqrt{ab}&1&\sqrt{cd}\\b&0&d\end{array}\right)
		\begin{pmatrix}0\\1\\1\end{pmatrix}=\begin{pmatrix}c\\1+\sqrt{cd}\\d\end{pmatrix}\]
		for some $a,b,c,d\in\mathbb{F}_q$ satisfying $ad+bc=1$. Suppose that $\theta$ stabilizes $\mathbf{g}_{q+3}$. Then $c=0,d=1$.  Since $ad+bc=1$, we have $a=1$. It follows from $b\in \mathbb{F}_q$ that $|M_{\mathbf{g}_{q+3}}|=q$, and thus $|\mathbf{g}_{q+3}^M|=q^2-1$. By Corollary \ref{equivalent1}, we conclude  that all NMDS codes in  $\mathscr{C}$ are monomially equivalent. This completes the proof.
	\end{proof}
	\begin{remark}
		For  regular o-polynomial, by setting $\mathbf{g}_{q+3}$ as $(0,1,1)^{\mathrm{T}}$  and $(1,1,0)^{\mathrm{T}}$ in \eqref{generator}, we obtain two distinct matrices $G_1$ and $G_2$. Two $[q+3,3,q]$ NMDS codes $\mathcal{C}_1$ and $\mathcal{C}_2$  with generator matrices $G_1$ and $G_2$, were constructed in \cite[Theorem 8]{Wang2021} and \cite[Theorem 10]{Li2023}, respectively. By Theorem \ref{re}, we see that $\mathcal{C}_1$ and $\mathcal{C}_2$ are monomially equivalent.
	\end{remark}
	Two $[q+1,3,q-2]$ NMDS codes $\mathcal{C}_3$ and $\mathcal{C}_4$ with the same weight enumerator were provided in \cite[Theorem 10]{Wang2021}  and \cite[Theorem 13]{Li2023} respectively, where $q=2^m$ and $m\geq3$ is odd. The generator matrices $G_i$  of $\mathcal{C}_i$, for $i=3,4$, are as follows:
	\[ G_3=\begin{pmatrix}1&1&\cdots&1&1&1\\\alpha_1&\alpha_2&\cdots&\alpha_{q-1}&1&0\\\alpha^2_1&\alpha^2_2&\cdots&\alpha^2_{q-1}&0&1\end{pmatrix}, \; G_4=\begin{pmatrix}1&1&\cdots&1&1&0\\\alpha_1&\alpha_2&\cdots&\alpha_{q-1}&1&1\\\alpha^2_1&\alpha^2_2&\cdots&\alpha^2_{q-1}&0&1\end{pmatrix},\]
	where $\mathbb{F}^{*}_{q}=\mathbb{F}_{q}\setminus\{0\}=\{\alpha_1,\dots,\alpha_{q-1}\}$.
	\begin{thm}
		Linear codes $\mathcal{C}_3$ and $\mathcal{C}_4$ are  not monomially equivalent.
	\end{thm}
	\begin{proof}
		Note that the first $q$ columns of $G_3$ and $G_4$ are  identical and  the first $q-1$ columns form a $(q-1)$-arc $\mathcal{A}=\mathcal{HO}\setminus\mathcal{S}$, where $\mathcal{HO}$ is a regular hyperoval and $\mathcal{S}=\{(1,0,0)^{\mathrm{T}},(0,0,1)^{\mathrm{T}},(0,1,0)^{\mathrm{T}}\}$.  Let $X$ and $Y$ denote the sets of points $\{(1,1,0)^{\mathrm{T}},(1,0,1)^{\mathrm{T}}\}$ and $\{(1,1,0)^{\mathrm{T}},(0,1,1)^{\mathrm{T}}\}$, respectively.  Suppose that $\mathcal{C}_3$ and $\mathcal{C}_4$ are monomially equivalent.
		Similarly  to the proof of Proposition \ref{equivalent}, there exists a homography $\theta$ such that $\mathcal{A}^{\theta}\cup X^{\theta}=\mathcal{A}\cup Y$.  Note that \[|\mathcal{HO}\cap\mathcal{HO}^{\theta}|=|(\mathcal{A}\cup\mathcal{S})\cap(\mathcal{A}^{\theta}\cup\mathcal{S}^{\theta})|\geq|\mathcal{A}\cap\mathcal{A}^{\theta}|\geq q-5>(q+2)/2\] for $q\geq32$. By Lemma \ref{unique}, this implies that   $\mathcal{HO}=\mathcal{HO}^{\theta}$, and consequently $\mathcal{A}=\mathcal{A}^{\theta}, \mathcal{S}=\mathcal{S}^{\theta}$. By Lemma \ref{regular}, the matrix representing $\theta$ 
		is of the form $\left(\begin{array}{ccc}a&0&c\\\sqrt{ab}&1&\sqrt{cd}\\b&0&d\end{array}\right)$, where $ad+bc=1$. Since $\mathcal{S}^{\theta}=\mathcal{S}$, we have \[\{(1,0,0)^{\mathrm{T}},(0,0,1)^{\mathrm{T}},(0,1,0)^{\mathrm{T}}\}=\{(a,\sqrt{ab},b)^{\mathrm{T}},(c,\sqrt{cd},d)^{\mathrm{T}},(0,1,0)^{\mathrm{T}}\}.\] It follows that  either $a=d=1,b=c=0$ or $a=d=0,b=c=1$. In either case, $\theta$ maps $(1,0,1)^{\mathrm{T}}$ to $(a+c,\sqrt{ab}+\sqrt{cd},b+d)=(1,0,1)^{\mathrm{T}}\notin Y$, which is a contradiction.
		
		Next, we consider the case $q=8$. Note that if $\mathcal{C}_3$ and $\mathcal{C}_4$ are  equivalent, then their extended codes generated by matrices  obtained by adding the same columns must also be equivalent. Let $\omega$ be a primitive element of $\mathbb{F}^{*}_8$. We added the column $(1,\omega,1)^{\mathrm{T}}$ to $G_3$ and $G_4$ to generate two linear codes $\mathcal{C}'_3$ and $\mathcal{C}'_4$, respectively. Using Magma \cite{Bosma}, we verify that $\mathcal{C}'_3$ is a $[10,3,7]$ linear code with weight enumerator
		\[A(z)=1+63z^7+126z^8+189z^9+133z^{10}\] and $\mathcal{C}'_4$ is a $[10,3,7]$ linear code with weight enumerator \[A(z)=1+56z^7+147z^8+168z^9+140z^{10}.\] Since the weight enumerators of $\mathcal{C}'_3$ and $\mathcal{C}'_4$  differ, $\mathcal{C}'_3$ and $\mathcal{C}'_4$ are not equivalent. This implies that $\mathcal{C}_3$ and $\mathcal{C}_4$ are also not equivalent. This completes the proof.   	
	\end{proof}
	Note that every hyperoval is regular for $q=2,4,8$.   From now on, we suppose that $q=2^m, m\geq4$.
	\begin{lemma}\label{irregular}
		\cite[Theorem 4.3]{O'Keefe1994}\cite[Theorem 4.11.6]{Payne} An irregular translation hyperoval with o-polynomial $f(x)=x^{\alpha}$ in $\mathrm{PG}(2,q)$, where $\alpha=2^h, q=2^m\geq16$, and $\gcd(h,m)=1$, has homography stabilizer of order $q(q-1)$. This  stabilizer is isomorphic to $\mathrm{AGL}(1,q)$ and consists of the maps:  $(x,y,z)^{\mathrm{T}}\mapsto A(x, y,z )^{\mathrm{T}}$, where 
		$$\left.A=\left(\begin{array}{ccc}1&0&0\\a&b&0\\a^{\alpha}&0&b^{\alpha}\end{array}\right.\right),\mathrm{~with~} a,b\in \mathbb{F}_q, b\neq0.$$
	\end{lemma} 
	\begin{thm}
		Let $\mathscr{C}$ be the set of NMDS codes over $\mathbb{F}_{q}$ with generator matrices of the form  \eqref{generator}. If $f(x)=x^{\alpha}$ is an irregular translation o-polynomial, where $\alpha=2^h, q=2^m\geq16$, and  $\gcd(h,m)=1$, then there are two equivalence classes in $\mathscr{C}$.
	\end{thm}
	\begin{proof}
		Similarly  to the proof of Theorem \ref{re}, let   $\mathcal{C}_{x,k}$ be the NMDS code in $\mathscr{C}$ with $\mathbf{g}_{q+3}=(1,x,x^\alpha+k)^{\mathrm{T}}$, where $x\in \mathbb{F}_q$ and $k\in \mathbb{F}^{*}_q$. 
		By Lemma \ref{irregular}, for any $\theta\in M$, we have 
		\[\mathbf{g}_{q+3}^{\theta}=\left(\begin{array}{ccc}1&0&0\\a&b&0\\a^{\alpha}&0&b^{\alpha}\end{array}\right)
		\begin{pmatrix}1\\x\\x^\alpha+k\end{pmatrix}=\begin{pmatrix}1\\a+bx\\(a+bx)^{\alpha}+kb^{\alpha}\end{pmatrix}\]
		for some $a\in \mathbb{F}_q$ and  $b\in \mathbb{F}^{*}_q$.	  Suppose that $\theta$ stabilizes $\mathbf{g}_{q+3}$. Then $a+bx=x$ and $b^{\alpha}=1$.  Since $f(x)=x^{\alpha}$ is a permutation polynomial over $\mathbb{F}_q$, it follows  that $a=0,b=1$.  Consequently, $|M_{\mathbf{g}_{q+3}}|=1$ and so $s_1:=|\mathbf{g}_{q+3}^M|=q^2-q$. 
		
		Let   $\mathcal{C}^{u}$ be the NMDS code in $\mathscr{C}$ with $\mathbf{g}_{q+3}=(0,1,u)^{\mathrm{T}}$, where $u\in \mathbb{F}^{*}_q$. 
		By Lemma \ref{irregular}, for any $\theta\in M$, we have 
		\[\mathbf{g}_{q+3}^{\theta}=\left(\begin{array}{ccc}1&0&0\\a&b&0\\a^{\alpha}&0&b^{\alpha}\end{array}\right)
		\begin{pmatrix}0\\1\\u\end{pmatrix}=\begin{pmatrix}0\\b\\ub^{\alpha}\end{pmatrix}\]  for some $a\in \mathbb{F}_q$ and  $b\in \mathbb{F}^{*}_q$.
		Suppose that $\theta$ stabilizes $\mathbf{g}_{q+3}$. Then $b^{2^h-1}=b^{\alpha-1}=1$. Since $b^{q-1}=b^{2^m-1}=1$, we have $b^{\gcd(2^h-1,2^m-1)}=1$ and so $b=1$. It follows  from $a\in\mathbb{F}_q$ that $|M_{\mathbf{g}_{q+3}}|=q$ and so $s_2:=|\mathbf{g}_{q+3}^M|=q-1$. Since $s_1+s_2=q^2-1$, by Corollary \ref{equivalent1}, we conclude that there are two equivalence classes in $\mathscr{C}$.  This completes the proof.
	\end{proof}
	
	\begin{lemma}\label{monomial}
		\cite[Theorem 4.4]{O'Keefe1994}\cite[Theorem 4.12.8]{Payne} Let $\mathcal{HO}$ be a  hyperoval with  o-polynomial $f(x)=x^{n}$, but  not a translation hyperoval. If $n^2-n+1 \not\equiv 0 \pmod{q-1}$, then $\mathcal{HO}$
		has homography stabilizer $H$ of order $(q-1)$  which consists of the maps:  $(x,y,z)^{\mathrm{T}}\mapsto A(x, y,z )^{\mathrm{T}}$, where 
		$$\left.A=\left(\begin{array}{ccc}1&0&0\\0&a&0\\0&0&a^n\end{array}\right.\right),\mathrm{~with~} a\in \mathbb{F}^{*}_q. $$
		If $n^2-n+1 \equiv 0 \pmod{q-1}$, then $H$ contains an additional homography $\sigma$ such that $|H|=3(q-1)$ and  the matrix representing $\sigma$ 
		is of the form 
		$\left(\begin{array}{ccc}0&1&0\\0&0&1\\1&0&0\end{array}\right)$. 
	\end{lemma} 
	
	\begin{thm}
		Let $\mathscr{C}$ be the set of NMDS codes over $\mathbb{F}_{q}$ with  generator matrices of the form  \eqref{generator}. If $f(x)=x^{n}$ is an o-polynomial  but not a translation o-polynomial, then there are $(q+1)/3$ equivalence classes in $\mathscr{C}$ if  $n^2-n+1 \equiv 0 \pmod{q-1}$. Otherwise, there are $q+r$ equivalence classes in $\mathscr{C}$, where $r:=\#\{a\in \mathbb{F}^{*}_q: a^{n-1}=1\}$.
	\end{thm}
	\begin{proof}
		We first consider the case where $n^2-n+1 \not\equiv 0 \pmod{q-1}$. 
		Similarly  to the proof of Theorem \ref{re}, let   $\mathcal{C}_{x,k}$ be the NMDS code in $\mathscr{C}$ with $\mathbf{g}_{q+3}=(1,x,x^n+k)^{\mathrm{T}}$, where  $x\in \mathbb{F}_q$ and $k\in \mathbb{F}^{*}_q$. 
		By Lemma \ref{monomial}, for any $\theta_a\in M$, we have 
		\[\mathbf{g}_{q+3}^{\theta_a}=\left(\begin{array}{ccc}1&0&0\\0&a&0\\0&0&a^n\end{array}\right)
		\begin{pmatrix}1\\x\\x^n+k\end{pmatrix}=\begin{pmatrix}1\\ax\\(ax)^n+a^nk\end{pmatrix}\]
		for some $a\in \mathbb{F}^{*}_q$. Suppose that $\theta_a$ stabilizes $\mathbf{g}_{q+3}$. Then $a=1$. It follows that $|M_{\mathbf{g}_{q+3}}|=1$, and thus $|\mathbf{g}_{q+3}^M|=q-1$. Hence, by Corollary \ref{equivalent1}, there are $q$ equivalence classes of NMDS codes in $\mathscr{C}$ such that  $\mathbf{g}_{q+3}$ is of the form  $(1,x,x^n+k)^{\mathrm{T}}$,  where  $x\in \mathbb{F}_q$ and $k\in \mathbb{F}^{*}_q$.
		Let   $\mathcal{C}^{u}$ be the NMDS code in $\mathscr{C}$ with $\mathbf{g}_{q+3}=(0,1,u)^{\mathrm{T}}$, where $u\in \mathbb{F}^{*}_q$. 
		By Lemma \ref{monomial}, for any $\theta_a\in M$, we have 
		\[\mathbf{g}_{q+3}^{\theta}=\left(\begin{array}{ccc}1&0&0\\0&a&0\\0&0&a^n\end{array}\right)
		\begin{pmatrix}0\\1\\u\end{pmatrix}=\begin{pmatrix}0\\a\\ua^{n}\end{pmatrix}\]
		for some $a\in \mathbb{F}^{*}_q$. Suppose that $\theta_a$ stabilizes $\mathbf{g}_{q+3}$. Then $a^{n-1}=1$. Let $r:=\#\{a\in \mathbb{F}^{*}_q: a^{n-1}=1\}$. It follows that $|M_{\mathbf{g}_{q+3}}|=r$ and so $|\mathbf{g}_{q+3}^M|=(q-1)/r$. By Corollary \ref{equivalent1}, there are $r$ equivalence classes of NMDS codes in $\mathscr{C}$ such that  $\mathbf{g}_{q+3}$ is of the form  $(0,1,u)^{\mathrm{T}}$, where $u\in \mathbb{F}^{*}_q$.   Hence, there are $q+r$ equivalence classes in $\mathscr{C}$.
		
		Next, we  consider the case where $n^2-n+1  \equiv 0 \pmod{q-1}$. By Lemma \ref{monomial}, there exists a homography $\sigma$ represented by the matrix $\left(\begin{array}{ccc}0&1&0\\0&0&1\\1&0&0\end{array}\right)$. We claim that $\sigma,\sigma\theta_a$ and $\sigma^2\theta_a$ cannot stabilize  any point not in  $\mathcal{HO}$ for any $a\in\mathbb{F}^{*}_q$. We only provide the details for  $\sigma\theta_a$ and the other cases are similar. It is obvious that the points of the form $(0,1,u)^{\mathrm{T}}$ cannot stabilized by $\sigma\theta_a$ for any $a\in\mathbb{F}^{*}_q$, where $u\in \mathbb{F}^{*}_q$. 
		Suppose that $\mathbf{g}_{q+3}=(1,x,x^n+k)^{\mathrm{T}}$ is stabilized by $\sigma\theta_a$ for some $a\in\mathbb{F}^{*}_q$, where $x$, $k$ and $x^n+k$ are all nonzero. Since \[\mathbf{g}_{q+3}^{\sigma\theta_a}=\left(\begin{array}{ccc}0&1&0\\0&0&1\\1&0&0\end{array}\right)\left(\begin{array}{ccc}1&0&0\\0&a&0\\0&0&a^n\end{array}\right)
		\begin{pmatrix}1\\x\\x^n+k\end{pmatrix}=\begin{pmatrix}ax\\a^n(x^n+k)\\1\end{pmatrix},\] we have $\frac{ax}{1}=\frac{a^n(x^n+k)}{x}=\frac{1}{x^n+k}$ and so $x^n+k=(ax)^{-1}$ and  $x^3=a^{n-2}$. From $n^2-n+1  \equiv 0 \pmod{q-1}$, it follows that $x^{3(n+1)}=a^{(n-2)(n+1)}=a^{-3}$. Since $\gcd(3,q-1)=1$, we obtain $a=x^{-(n+1)}$. Thus $x^n+k=(ax)^{-1}=x^n$ and so $k=0$, which is a contradiction. On the other hand, we have $\gcd(n-1,q-1)=1$ by $n^2-n+1  \equiv 0 \pmod{q-1}$. Thus $r:=\#\{a\in \mathbb{F}^{*}_q: a^{n-1}=1\}=1$.  Combining this with   the proof in the first paragraph, we conclude that there are   $(q+1)/3$ equivalence classes in $\mathscr{C}$.  This completes the proof.\end{proof}
	\begin{thm}
		Let $\mathscr{C}$ be the set of NMDS codes over $\mathbb{F}_{q}$ with  generator matrices of the form  \eqref{generator}. If $f(x)$ is the Cherowitzo o-polynomial, then there  are $q^2-1$ equivalence classes in $\mathscr{C}$.
	\end{thm}  
	\begin{proof}
		By \cite[Corollary 4.4]{Bayens2007}, the Cherowitzo hyperoval has trivial homography stabilizer. The claim  then follows from Corollary \ref{equivalent1}.
	\end{proof}
	\begin{thm}\label{Payneoval}
		Let $\mathscr{C}$ be the set of NMDS codes over $\mathbb{F}_{q}$ with  generator matrices of the form  \eqref{generator}. If $f(x)$ is a Payne o-polynomial or  Adelaide o-polynomial, then there are $(q^2+q-2)/2$ equivalence classes in $\mathscr{C}$.
	\end{thm}  
	\begin{proof}
		By \cite{Thas1988}, the Payne hyperoval has homography stabilizer of order $2$, generated by the following matrix   	$\left(\begin{array}{ccc}0&1&0\\1&0&0\\0&0&1\end{array}\right)$. This stabilizer fixes  points of the form $(1,1,m)^{\mathrm{T}}$, where $m\in\mathbb{F}_{q}\setminus\{1\}$, and exchanges the remaining $q^2-q$ points not in the hyperoval  pairwise. The claim then follows from Corollary \ref{equivalent1}.
		
		By \cite{Payne2005}, the Adelaide hyperoval has homography stabilizer of order $2$, generated by the following matrix  	$\left(\begin{array}{ccc}1&0&0\\\beta+\beta^{-1}&1&0\\0&0&1\end{array}\right)$, where $\beta\in\mathbb{F}_{q^2}\setminus\{1\}$ satisfies $\beta^{q+1}=1$. This stabilizer fixes   points of the form $(0,1,m)^{\mathrm{T}}$, where $m\in\mathbb{F}^{*}_{q}$ and exchanges the remaining $q^2-q$ points not in the hyperoval  pairwise. The claim then follows from Corollary \ref{equivalent1}.
	\end{proof}
	For small $q$, each Subiaco hyperoval  falls into  one of the previously known classes of hyperovals. If $q=16$,  a Subiaco hyperoval is a Lunelli-Sce hyperoval \cite{Lunelli1958}. Its homography stabilizer has order $36$ and is isomorphic to $\mathrm{C}^{3}_{2}\rtimes \mathrm{C}_4$ \cite{Payne1978}. 	If $q=32$,  a Subiaco hyperoval is a Payne hyperoval \cite{Payne1985}. If $q=64$,  the two projectively distinct Subiaco hyperovals are the two Penttila-Pinneri irregular hyperovals \cite{Penttila1994}. The homography stabilizer is either $\mathrm{C}_5$ or $\mathrm{D}_{10}$. If $q=2^m>64$ with $m\not\equiv 2 \pmod{4}$,   the homography
	stabilizer of a Subiaco hyperoval is  a cyclic group of order $2$, similar to the case $q=32$ \cite{O'Keefe1996}. If $q=2^m>64$ with $m\equiv 2 \pmod{4}$,  the homography
	stabilizer of a Subiaco hyperoval is either $\mathrm{C}_5$ or $\mathrm{D}_{10}$, similar to the case $q=64$ \cite{Payne1995}. We  provide details for only  one case.
	\begin{thm}
		Let $\mathscr{C}$ be the set of NMDS codes over $\mathbb{F}_{q}$ with  generator matrices of the form  \eqref{generator}, where $q=2^m>64$ with $m\equiv 2 \pmod{4}$. If $f(x)=\frac{\delta^2(x^4+x)}{\left(x^2+\delta x+1\right)^2}+x^{1/2}$ is a Subiaco o-polynomial with $\delta^2+\delta+1=0$, then there are $(q^2-1)/5$ equivalence classes in $\mathscr{C}$.
	\end{thm}  
	\begin{proof}
		In this case, the  Subiaco hyperoval has homography stabilizer $\mathrm{C}_5$, generated by $\theta: \begin{aligned}(x,y,z)^{\mathrm{T}}\mapsto(y,x+\delta y,z+\delta^2 y)^{\mathrm{T}}\end{aligned}$.  We can verify that  the points not in the Subiaco hyperoval are not stabilized by $\theta$. The claim then  follows from Corollary \ref{equivalent1}. 
	\end{proof}
	
	\begin{thm}
		Let $\mathscr{C}$ be the set of NMDS codes over $\mathbb{F}_{q}$ with  generator matrices of the form  \eqref{generator}. If $f(x)$ is the  O'Keefe-Penttila o-polynomial, then there are $341$ equivalence classes in $\mathscr{C}$.
	\end{thm}  
	\begin{proof}
		By \cite[Theroem 3.2]{O'Keefe1992}, the  O'Keefe-Penttila hyperoval has homography stabilizer of order $3$, generated by the  matrix 	$\left(\begin{array}{ccc}0&1&0\\1&1&1\\0&0&1\end{array}\right)$. It cannot fix any point not in  the  hyperoval. The  claim then follows from Corollary \ref{equivalent1}.
	\end{proof}
	\section{NMDS codes of dimension $4$}\label{s4}
	Recall that the point-set $$\mathcal{S}_h=\{P_x=(1,x,x^{2^h},x^{2^h+1})\mid x\in\mathbb{F}_q\}\cup\{P_{\infty}=(0,0,0,1)\}$$
	is  a $(q+1)$-arc in $\mathrm{PG}(3,q)$ for some prime power $q=p^m$, where $\gcd(m,h)=1$.
	Let the columns of matrix $G$ be constructed by the $q+1$ points of  $\mathcal{S}_h$ and one point not in $\mathcal{S}_h$. Then $G$ has the form:
	\begin{equation}\label{generator1}
		\begin{aligned}G=(\mathbf{g}_1,\ldots,\mathbf{g}_{q+1},\mathbf{g}_{q+2})=\begin{pmatrix}1&1&\cdots&1&1&0&a\\\alpha_1&\alpha_2&\cdots&\alpha_{q-1}&0&0&b\\\alpha_1^{2^h}&\alpha_2^{2^h}&\cdots&\alpha_{q-1}^{2^h}&0&0&c\\\alpha_1^{2^h+1}&\alpha_2^{2^h+1}&\cdots&\alpha_{q-1}^{2^h+1}&0&1&d\end{pmatrix},\end{aligned}
	\end{equation}
	where $\mathbb{F}^{*}_{q}=\mathbb{F}_{q}\setminus\{0\}=\{\alpha_1,\dots,\alpha_{q-1}\}$ and   $\mathbf{g}_{q+2}=(a,b,c,d)^{\mathrm{T}}\neq\mathbf{0}$. Let $\mathcal{S}_{G}$ denote the set of points $\{\mathbf{g}_1,\ldots,\mathbf{g}_{q+1},\mathbf{g}_{q+2}\}$.
	\subsection{ $q$ even}
	In this subsection, we suppose that $q=2^m\geq8$  such that $\gcd(h,m)=1$. Let $K_h$ be the stabilizer of $\mathcal{S}_h$ in  $\mathrm{PGL}(4,q)$. Then $K_h$ is isomorphic to $\mathrm{PGL}(2,q)$, and its elements are represented by the matrices 
	$$
	M_{a, b, c, d}=\left(\begin{array}{cccc}
		a^{2^h+1} & a^{2^h}c & a c^{2^h}&c^{2^h+1}   \\
		a^{2^h}b & a^{2^h} d &b c^{2^h}   &c^{2^h} d \\
		ab^{2^h} & b^{2^h} c& a d^{2^h} & c d^{2^h} \\
		b^{2^h+1}& b^{2^h} d  & b d^{2^h} & d^{2^h+1}
	\end{array}\right) \text {, }
	$$
	where $a,b,c,d\in\mathbb{F}_q$ satisfy $ad+bc=1$, see \cite[Lemma 21.3.13, Theorem 21.3.17]{Hirschfeld1985}. We now present  the  results on the orbits of points and planes under the action of $K_h$.
	\begin{lemma}\label{points_orbits}
		\cite[Proposition 3.2]{Ceria2023} The group $K_h$ has five orbits on the points of $\mathrm{PG}(3,q)$:
		\begin{enumerate}
			\item[(1)] $\mathcal{M}_1:= \mathcal{S}_h$ of size $q+1$;
			\item[(2)] $\mathcal{M}_2$ of size $q^2+q$, consisting of points on exactly two osculating planes;
			\item[(3)] $\mathcal{M}_3$ of size $(q^3-q)/6$, consisting of points on exactly three osculating planes;
			\item[(4)]  $\mathcal{M}_4$ of size $(q^3-q)/2$, consisting of points on exactly one osculating planes;
			\item[(5)]  $\mathcal{M}_5$ of size $(q^3-q)/3$, consisting of points not contained in any osculating planes.
		\end{enumerate}
		Let	$R$  be the set of points  on real chords but not in $\mathcal{S}_h$. If $q \equiv -1 \pmod{3}$, then $R=\mathcal{M}_4$; if $q \equiv 1 \pmod{3}$, then $R=\mathcal{M}_3\cup\mathcal{M}_5$. 
	\end{lemma}
	Recall that the osculating planes of $\mathcal{S}_h$  at $P_x$ and $P_{\infty}$ are defined by the equations $x^{2^h+1}X_1+x^{2^h}X_2+xX_3+X_4=0$ and $X_1=0$, respectively. Let $\pi(P)$ denote the osculating planes of $\mathcal{S}_h$ at $P$ and $\pi(\mathcal{M}_i)=\{\pi(P)| P\in\mathcal{M}_i\}$, where  $i=1,\dots,5$. The $K_h$-orbits on the  planes  of $\mathrm{PG}(3,q)$ can be described by the following lemma.
	\begin{lemma}\label{planes_orbits}
		\cite[Proposition 3.4]{Ceria2023} The group $K_h$ has five orbits on the  planes of $\mathrm{PG}(3,q)$:
		\begin{enumerate}
			\item[(1)] $\pi(\mathcal{M}_1)$ consisting of the $q+1$ osculating planes;
			\item[(2)] $\pi(\mathcal{M}_2)$ of size $q^2+q$, consisting of  planes containing   exactly two points of $\mathcal{S}_h$;
			\item[(3)] $\pi(\mathcal{M}_3)$ of size $(q^3-q)/6$, consisting of  planes containing  exactly three points of $\mathcal{S}_h$;
			\item[(4)]  $\pi(\mathcal{M}_4)$ of size $(q^3-q)/2$, consisting of  planes containing  exactly one point of $\mathcal{S}_h$;
			\item[(5)]  $\pi(\mathcal{M}_5)$ of size $(q^3-q)/3$, consisting of  planes containing  no point of $\mathcal{S}_h$.
		\end{enumerate}
	\end{lemma}
	Recall that  a $[q+2,4,q-2]$ NMDS code can be viewed as a point-set $\mathcal{S}$ of size $q+2$ in $\mathrm{PG}(3,q)$  satisfying the following conditions:  
	\begin{enumerate}
		\item[(1)] any three points of $\mathcal{S}$ generate a plane in $\mathrm{PG}(3,q)$;
		\item[(2)] there exist four points in $\mathcal{S}$ that lie on a plane in $\mathrm{PG}(3,q)$;
		\item[(3)] every five points of $\mathcal{S}$ generate $\mathrm{PG}(3,q)$.
	\end{enumerate}
	Let $\mathcal{C}$ be the linear code over $\mathbb{F}_{q}$ with the generator matrix $G$ defined in \eqref{generator1}. Suppose that the point $\mathbf{g}_{q+2}=(a,b,c,d)^{\mathrm{T}}$ does not lie on any   real chord. Then any three columns of matrix $G$ are linearly independent,  ensuring that condition (1) holds. Since the   first $q+1$ columns of $G$ constitute the $(q+1)$-arc $\mathcal{S}_h$,  we have $|\mathcal{H}_{u}\cap\mathcal{S}_h|\leq3$ for any plane $\mathcal{H}_{u}$. It follows that $|\mathcal{H}_{u}\cap \mathcal{S}_{G}|\leq4$, which implies that  condition (3) holds. We now verify condition (2) for different cases and provide some equivalence classes of NMDS codes.
	\begin{thm}\label{4_even}
		Let $\mathscr{C}$ be the set of linear codes over $\mathbb{F}_{q}$ with  generator matrices of the form  \eqref{generator1}. Then there are either  two or three  equivalence classes of NMDS codes in $\mathscr{C}$ depending on  whether $q \equiv 1 \pmod{3}$ or $q \equiv -1 \pmod{3}$. Moreover, the weight distributions of these codes are provided in Table \ref{enumerator1}.
	\end{thm}
	\begin{proof}
		By Lemma \ref{points_orbits}, if $q \equiv -1 \pmod{3}$ and  $\mathbf{g}_{q+2}\in\mathcal{M}_4$  or    $q \equiv 1 \pmod{3}$ and $\mathbf{g}_{q+2}\in\mathcal{M}_3\cup\mathcal{M}_5$,  then the corresponding linear code is not an NMDS code. For the remaining cases,	we only need consider the number of nonzero vectors $u\in\mathbb{F}^4_q$ such that  $|\mathcal{H}_{u}\cap \mathcal{S}_{G}|=4$. Suppose that a nonzero vector $u\in\mathbb{F}^4_q$ satisfies this condition. Then $\mathbf{g}_{q+2}\in \mathcal{H}_{u}$ and so $|\mathcal{H}_{u}\cap\mathcal{S}_h|=3$. By Lemma \ref{planes_orbits}, all planes containing exactly three points of $\mathcal{S}_h$ form a $K_h$-orbit  $\pi(\mathcal{M}_3)$. This means the number of nonzero vectors $u\in\mathbb{F}^4_q$ such that  $|\mathcal{H}_{u}\cap\mathcal{S}_G|=4$ is equal to $q-1$ times the number of planes in $\pi(\mathcal{M}_3)$ through the point $\mathbf{g}_{q+2}$. By  \cite[Table 1]{Ceria2023}, these numbers are all great than 0 and only depend  on the $K_h$-orbit on points to which $\mathbf{g}_{q+2}$ belongs.   Thus, the remaining linear codes are all NMDS codes. Obviously, the codes corresponding to $\mathbf{g}_{q+2}$  in the same $K_h$-orbit are monomially equivalent, while the codes corresponding to $\mathbf{g}_{q+2}$ in different $K_h$-orbits are not monomially equivalent as they have different numbers of minimum weight codes.  Therefore, there are two or three equivalence classes of NMDS codes in $\mathscr{C}$ depending on whether $q \equiv 1 \pmod{3}$ or $q \equiv -1 \pmod{3}$. Finally, by Lemma \ref{enumerators}, the weight distributions of these codes can be fully determined and we provide them in Table \ref{enumerator1}.
	\end{proof}
	\begin{table}
		\begin{center}
			\caption{The weight distributions of NMDS codes in Theorem \ref{4_even} when $\mathbf{g}_{q+2}\in \mathcal{M}_i$.  }\label{enumerator1}
			\resizebox{\textwidth}{!}{
				\begin{tabular}{ccccc}
					\hline  & $\mathcal{M}_2$ & $\mathcal{M}_3$ &$\mathcal{M}_4$ &$\mathcal{M}_5$\\
					\hline
					$A_{q-2}$ & $(q-1)^2(q-2)/6$ &$(q-1)(q^2-q+4)/6$  &  $q(q-1)^2/6$& $(q^2-1)(q-2)/6$ \\
					$A_{q-1}$ & $(q-1)(q^3-q^2+14q-8)/6$ &$(q-1)(q-2)(q^2+q+8)/6$ &$q(q-1)(q^2-q+6)/6$  &$(q^2-1)(q^2-2q+8)/6$  \\
					$A_{q}$ & $3(q-1)(q^3-q+2)/2$ & $(q-1)(3q^2+q+10)/2$ & $(q-1)(3q^2+q+2)/2$&  $(q^2-1)(3q-2)/2$\\  
					$A_{q+1}$ & $(q-1)^2(3q^2+2q+8)/6$ &$(q-1)(q-2)(3q^2+5q+8)/6$  &$q(q-1)^2(3q+2)/6$ & $(q-1)(3q^3-q^2-2q+8)/6$ \\  
					$A_{q+2}$ & $(q-1)^3(q+1)/3$ & $(q^2-1)(q^2-2q+2)/3$ &$q^2(q-1)^2/3$ &$(q-1)(q^3-q^2-1)/3$  \\
					Condition & $q \equiv \pm1\pmod{3}$  & $q \equiv -1\pmod{3}$ &$q \equiv 1\pmod{3}$ & $q \equiv -1\pmod{3}$ \\
					\hline
			\end{tabular}}
		\end{center}
	\end{table}
	\begin{remark}
		Let $\mathbf{g}^{\mathrm{T}}_{q+2}=(0,0,1,0)\in\mathcal{M}_2$ in \eqref{generator1}. Then we can directly obtain the corresponding result from \cite[Theorem 15]{Ding2024} using Theorem \ref{4_even}.
	\end{remark}
	
	\begin{example}
		Let $q=8,h=1$,  and take $\mathbf{g}^{\mathrm{T}}_{q+2}$ as $(0,1,0,0),(1,0,1,1),(1,\zeta,1,1)$ respectively, where $\zeta$ is  a generator of $\mathbb{F}^{*}_8$. Using Magma \cite{Bosma}, we obtain three $[10,4,6]$ NMDS codes with the following weight enumerators:
		\[
		A(z)=1+49z^6+644z^7+609z^8+1764z^9+1029z^{10},
		\]
		\[
		A(z)=1+70z^6+560z^7+735z^8+1680z^9+1050z^{10},
		\] 
		\[
		A(z)=1+63z^6+588z^7+693z^8+1708z^9+1043z^{10},
		\]
		respectively, which coincides with the conclusion of Theorem \ref{4_even}.
		
		Let $q=16,h=1$, and take $\mathbf{g}^{\mathrm{T}}_{q+2}$ as $(0,1,0,0),(1,\zeta^3,1,1)$ respectively, where $\zeta$ is a generator of $\mathbb{F}^{*}_{16}$. Using Magma \cite{Bosma}, we obtain two $[18,4,14]$ NMDS codes with the following weight enumerators:
		\[
		A(z)=1+525z^{14}+10140z^{15}+5445z^{16}+30300z^{17}+19125z^{18}, 
		\]
		\[
		A(z)=1+600z^{14}+9840z^{15}+5895z^{16}+30000z^{17}+19200z^{18},
		\] 
		respectively, which coincides with the conclusion of Theorem \ref{4_even}.
	\end{example}

	\subsection{ $q$ odd}
	In this subsection, we consider the case with odd $q$. Recall that if $q$ is odd every $(q+1)$-arc is a twisted cubic. Thus, we suppose that $(q+1)$-arc is the point-set
	$$\mathcal{S}_1=\{(1,x,x^{2},x^{3})\mid x\in\mathbb{F}_q\}\cup\{(0,0,0,1)\}.$$
	Similarly to the even case, let $K_q$ be the stabilizer of $\mathcal{S}_1$ in  $\mathrm{PGL}(4,q)$. Then for $q>4$,  $K_q$ is also isomorphic to $\mathrm{PGL}(2,q)$, 
	and its elements are represented by the matrices 
	$$
	M_{a, b, c, d}=\left(\begin{array}{cccc}
		a^{3} & 3a^{2}c & 3ac^{2}&c^{3}   \\
		a^{2}b & a^{2}d+2abc & bc^{2}+2acd   &c^{2}d \\
		ab^{2} & b^{2}c+2abd & ad^{2}+2bcd & cd^{2} \\
		b^{2}& 3b^{2}d  & 3bd^{2} & d^{3}
	\end{array}\right) \text {, }
	$$
	where $a,b,c,d\in\mathbb{F}_q$ satisfy $ad-bc\neq0$, 
	as shown in \cite[Lemma 21.1.3]{Hirschfeld1985}. Under $K_q$, there are five orbits of points and planes. The results on point-orbits (when $q\not\equiv 0\pmod{3}$) and plane-orbits  coincide with those in  the case of even $q$ (see \cite[Chapter 21]{Hirschfeld1985} or \cite[Theorem 2.2]{Bartoli2020}). In the case $q\equiv 0\pmod{3}$, we have the following result.
	\begin{lemma}\label{points_orbits3}
		\cite[Lemma 21.1.3]{Hirschfeld1985}\cite[Theorem 2.2]{Bartoli2020} Let $q\equiv 0\pmod{3}$. The group $K_q$ has five orbits on the points of $\mathrm{PG}(3,q)$:
		\begin{enumerate}
			\item[(1)] $\mathcal{M}_1:= \mathcal{S}_1$ of size $q+1$;
			\item[(2)] $\mathcal{M}_2$ of size $q+1$, consisting of points on all osculating planes;
			\item[(3)] $\mathcal{M}_3$ of size $q^2-1$, consisting of points on exactly one osculating plane;
			\item[(4)]  $\mathcal{M}_4$ of size $(q^3-q)/2$, consisting of points on  a real chord;
			\item[(5)]  $\mathcal{M}_5$ of size $(q^3-q)/2$, consisting of points on an imaginary chord.
		\end{enumerate}
	\end{lemma}
	\begin{thm}\label{4_odd}
		Let $\mathscr{C}$  be the set of linear codes over $\mathbb{F}_{q}$ with  generator matrices of the form  \eqref{generator1}. If $q\not\equiv 0\pmod{3}$, then there are two or three equivalence classes of NMDS codes in $\mathscr{C}$, depending on whether $q \equiv 1 \pmod{3}$ or $q \equiv -1 \pmod{3}$. Moreover, the weight distributions of these codes  are the same as those in Theorem \ref{4_even}. If $q\equiv 0\pmod{3}$, then there are three equivalence classes of NMDS codes in $\mathscr{C}$, and their  weight distributions   are given in Table \ref{enumerator2}.
	\end{thm}
	\begin{proof}
		Note that the number of   planes containing exactly three points of $\mathcal{S}_1$  through a point of  $\mathcal{M}_i$ can be found in \cite[Table 1, 2]{Bartoli2020}. To prove the theorem, it suffices to show that in the case $q\equiv 0\pmod{3}$, the  linear codes corresponding to  $\mathbf{g}_{q+2}$ in the $K_q$-orbits $\mathcal{M}_2$ and  $\mathcal{M}_5$ are not monomially equivalent.  The remaining proofs  follow similarly to those in Theorem \ref{4_even}. Let $\mathcal{C}_1$ and $\mathcal{C}_{2}$ be linear codes with generator matrices $G_1=(\mathbf{g}_1,\ldots,\mathbf{g}_{q+1},\mathbf{e})$ and $G_{2}=(\mathbf{g}_1,\ldots,\mathbf{g}_{q+1},\mathbf{f})$ respectively, where $\{\mathbf{g}_1,\ldots,\mathbf{g}_{q+1}\}$ forms the $(q+1)$-arc $\mathcal{S}_1$ and $\mathbf{e}\in\mathcal{M}_2, \mathbf{f}\in\mathcal{M}_5$. Suppose that $\mathcal{C}_1$ and $\mathcal{C}_{2}$  are monomially equivalent. Then there  is a  map $\theta\in\mathrm{PGL}(4,q)$ such that the sets of projective points $\{\mathbf{g}^{\theta}_1,\ldots,\mathbf{g}^{\theta}_{q+1},\mathbf{e}^{\theta}\}=\{\mathbf{g}_1,\ldots,\mathbf{g}_{q+1},\mathbf{f}\}$. Obviously, $\theta\notin K_q$   since $\mathbf{e}$ and $\mathbf{f}$ belongs to different $K_q$-orbits. Therefore, we must have $\mathbf{e}^{\theta}=\mathbf{g}_{i}$ for some $i\in\{1,\ldots,q+1\}$. By Lemma \ref{points_orbits3}, the point $\mathbf{e}$ lies on all osculating planes and so does $\mathbf{e}^{\theta}=\mathbf{g}_{i}$, which is impossible. Thus $\mathcal{C}_1$ and $\mathcal{C}_{2}$  are  not monomially equivalent. This completes the proof.
	\end{proof}
	\begin{table}
		\begin{center}
			\caption{The weight distributions of NMDS codes in Theorem \ref{4_odd} when $\mathbf{g}_{q+2}\in \mathcal{M}_i$ and $q\equiv 0\pmod{3}$.  }\label{enumerator2}
			\resizebox{\textwidth}{!}{
				\begin{tabular}{ccccc}
					\hline  & $\mathcal{M}_2$ & $\mathcal{M}_3$  &$\mathcal{M}_5$\\
					\hline
					$A_{q-2}$ & $q(q-1)^2/6$ &$q(q-1)(q-3)/6$  & $q(q-1)^2/6$ \\
					$A_{q-1}$ & $q(q-1)(q^2-q+6)/6$ &$q(q-1)(q^2-q+14)/6$& $q(q-1)(q^2-q+6)/6$  \\
					$A_{q}$ & $(q-1)(3q^2+q+2)/2$ & $(q-1)(3q^2-3q+2)/2$ & $(q-1)(3q^2+q+2)/2$\\  
					$A_{q+1}$ & $q(q-1)^2(3q+2)/6$ &$q(q-1)(3q^2-q+6)/6$  &$q(q-1)^2(3q+2)/6$ \\  
					$A_{q+2}$ & $q^2(q-1)^2/3$ & $q(q-1)(q^2-q-1)/3$ & $q^2(q-1)^2/3$  \\
					\hline
			\end{tabular}}
		\end{center}
	\end{table}
	\begin{remark}
		In \cite{Ceria}, the authors referred to NMDS codes as \emph{NMDS-sets}. An NMDS-set is said to be \emph{complete} if it is maximal with respect to set inclusion. They showed that extending $\mathcal{S}_1$ by adding a point on a tangent line or an imaginary chord yields an NMDS-set, which can be completed by adding at most one or three additional points, respectively.
	\end{remark}
	\subsection{NMDS codes with parameters $[q+3,4,q-1]$}
	Let $q=p^m$ be a prime power  and $h$ be a positive integer.
	Define $4\times(q+3)$ matrix $G_3$ by
	\begin{equation}\label{g3}
		\begin{aligned}G_3=(\mathbf{g}_1,\ldots,\mathbf{g}_{q+1},\mathbf{g}_{q+2},\mathbf{g}_{q+3})=\begin{pmatrix}1&1&\cdots&1&1&0&a_1&a_2\\\alpha_1&\alpha_2&\cdots&\alpha_{q-1}&0&0&b_1&b_2\\\alpha_1^{2^h}&\alpha_2^{2^h}&\cdots&\alpha_{q-1}^{2^h}&0&0&c_1&c_2\\\alpha_1^{2^h+1}&\alpha_2^{2^h+1}&\cdots&\alpha_{q-1}^{2^h+1}&0&1&d_1&d_2\end{pmatrix},\end{aligned}
	\end{equation}
	where $\mathbb{F}^{*}_{q}=\mathbb{F}_{q}\setminus\{0\}=\{\alpha_1,\dots,\alpha_{q-1}\}$ and  $\mathbf{g}_{q+2},\mathbf{g}_{q+3}\neq\mathbf{0}$.
	By selecting specific $\mathbf{g}_{q+2}$ and $\mathbf{g}_{q+3}$, we can construct certain NMDS codes. For example, let $p=2$ and $h$ be a positive integer satisfying $\gcd(m,h)=1$, where $m\geq3$ is odd. Take $\mathbf{g}_{q+2}=(0,0,1,0)^{\mathrm{T}}\in\mathcal{M}_2$ and $\mathbf{g}_{q+3}=(0,1,0,0)^{\mathrm{T}}\in\mathcal{M}_2$ or $\mathbf{g}_{q+3}=(0,1,1,0)^{\mathrm{T}}\in\mathcal{M}_3$, where $\mathcal{M}_2$ and $\mathcal{M}_3$ are the point-orbits of $K_h$ defined in Lemma \ref{points_orbits}. Then the linear code generated by $G_3$ is a $[q+3,4,q-1]$ NMDS code (see \cite[Theorems 18 and 19]{Ding2024}). A simpler proof of this result can be provided via Theorem \ref{4_even}.
	\begin{thm}\label{NMDS}
		Let $q=2^m$ and $h$ be a positive integer satisfying $\gcd(m,h)=1$. Take $\mathbf{g}_{q+2}=(0,0,1,0)^{\mathrm{T}}\in\mathcal{M}_2$ and $\mathbf{g}_{q+3}=(0,1,0,0)^{\mathrm{T}}\in\mathcal{M}_2$ in \eqref{g3}. Let $\mathcal{C}_3$  be the linear code generated by $G_3$. Then $\mathcal{C}_3$ is a $[q+3,4,q-1]$ NMDS code with $A_{q-1}=(q-1)^2(q-2)/3$ if and only if $m\geq3$ is odd.
	\end{thm}	 
	\begin{proof}
		First, we can verify  that any three columns of $G_3$ are linearly independent, i.e. $\mathcal{C}_3$ satisfy condition (1). Then, for any nonzero vector  $u=(u_1,u_2,u_3,u_4)\in\mathbb{F}^{4}_q$, we calculate  $|\mathcal{H}_{u}\cap \mathcal{S}_{G_3}|$, where $\mathcal{S}_{G_3}$ denotes the set of points $\{\mathbf{g}_1,\ldots,\mathbf{g}_{q+1},\mathbf{g}_{q+2},\mathbf{g}_{q+3}\}$. Suppose that $\mathbf{g}_{q+2}\in\mathcal{H}_{u}$ and $\mathbf{g}_{q+3}\notin\mathcal{H}_{u}$. By the proof of Theorem \ref{4_even}, we have   $|\mathcal{H}_{u}\cap \mathcal{S}_{G_3}|\leq4$ and the number of nonzero vectors $u$ satisfying this condition is $(q-1)^2(q-2)/6$ from Table \ref{enumerator1}. The case where $\mathbf{g}_{q+2}\notin\mathcal{H}_{u}$ and $\mathbf{g}_{q+3}\in\mathcal{H}_{u}$ also yields the same results.
		
		Finally, suppose that  $\mathbf{g}_{q+2},\mathbf{g}_{q+3}\in\mathcal{H}_{u}$. Then $u_2=u_3=0$. If $(1,x,x^{2^h},x^{2^h+1})\in\mathcal{H}_{u}$, then we have $u_1+u_4x^{2^h+1}=0$. If $m$ is odd, then $\gcd(2^h+1,q-1)=\gcd(2^h+1,2^m-1)=1$ since $\gcd(h,m)=1$. By Lemma \ref{root}, the equation $u_1+u_4x^{2^h+1}=0$ has exactly one solution, implying  $|\mathcal{H}_{u}\cap \mathcal{S}_{G_3}|=3$. Therefore,
		$|\mathcal{H}_{u}\cap \mathcal{S}_{G_3}|\leq4$ holds for any nonzero vector $u\in\mathbb{F}^{4}_q$, and the number of  vectors $u\in\mathbb{F}^{4}_q$ such that $|\mathcal{H}_{u}\cap \mathcal{S}_{G_3}|=4$ is $(q-1)^2(q-2)/6+(q-1)^2(q-2)/6=(q-1)^2(q-2)/3$. Thus,
		$\mathcal{C}_3$ is a   $[q+3,4,q-1]$ NMDS code with $A_{q-1}=(q-1)^2(q-2)/3$.  If $m$ is even, then $\gcd(2^h+1,q-1)=\gcd(2^h+1,2^m-1)=2^{\gcd(h,m)}+1=3$. By Lemma \ref{root}, the equation $u_1+u_4x^{2^h+1}=0$ has exactly three solutions, leading to   $|\mathcal{H}_{u}\cap \mathcal{S}_{G_3}|=5$. Hence, $\mathcal{C}_3$ is not an NMDS code.  This completes the proof.
	\end{proof}
	
	\begin{thm}
		Let $q=2^m$ and $h$ be a positive integer satisfying $\gcd(m,h)=1$. Take $\mathbf{g}_{q+2}=(0,0,1,0)^{\mathrm{T}}\in\mathcal{M}_2$ and $\mathbf{g}_{q+3}=(0,1,1,0)^{\mathrm{T}}\in\mathcal{M}_3$ in \eqref{g3}. Let $\mathcal{C}_4$ be the linear code generated by $G_3$. Then $\mathcal{C}_4$ is a $[q+3,4,q-1]$ NMDS code with $A_{q-1}=(q-1)(q^2-2q+3)/3$ if and only if $m\geq3$ is odd.
	\end{thm}
	\begin{proof}
		The proof follows the approach of Theorem \ref{NMDS}. The value of $A_{q-1}$ is obtained by summing the entries in the second column $(q-1)^2(q-2)/6$ and the third column $(q-1)(q^2-q+4)/6$ of the second row in Table \ref{enumerator1}. 
	\end{proof}	 
	
	Let $h=1$ and $\gcd(3,q-1)=1$. Taking $\mathbf{g}_{q+2}=(0,0,1,0)^{\mathrm{T}}\in\mathcal{M}_2$ and $\mathbf{g}_{q+3}=(0,1,0,0)^{\mathrm{T}}\in\mathcal{M}_2$ in \eqref{g3}, the linear code generated by $G_3$ is a $[q+3,4,q-1]$ NMDS code (see   \cite[Theorem 9]{Xu2023}).  
	Similarly  to the proof of Theorem \ref{NMDS}, we further have the following theorem from Tables \ref{enumerator1} and \ref{enumerator2}.
	
	\begin{thm}\label{NMDS2}
		Let  $h=1$ and $q=p^m$. Take $\mathbf{g}_{q+2}=(0,0,1,0)^{\mathrm{T}}\in\mathcal{M}_2$ and $\mathbf{g}_{q+3}=(0,1,1,0)^{\mathrm{T}}\in\mathcal{M}_2$ in \eqref{g3}. Let $\mathcal{C}_5$ be the linear code generated by $G_3$. Then $\mathcal{C}_5$ is a $[q+3,4,q-1]$ NMDS code if and only if  $\gcd(3,q-1)=1$. Moreover, $A_{q-1}=q(q-1)^2/3$ if $p=3$ and $A_{q-1}=(q-1)^2(q-2)/3$ if $p\neq3$.
	\end{thm}
	When $q=2^m$ with $m$ odd, Theorems \ref{NMDS} and \ref{NMDS2} show that $\mathcal{C}_3$ and $\mathcal{C}_5$ have identical weight enumerators by Lemma \ref{enumerators}. However, as noted in \cite[Remark 22]{Ding2024}, $\mathcal{C}_3$ and $\mathcal{C}_5$ are not monomially equivalent $[q+3,4,q-1]$ NMDS codes for $q=2^3,2^5,2^7$ and some possible values of $h$. Finally, we have the following proposition.
	
	\begin{prop}
		If $q=2^m$ with $m\geq7$ and $h\geq2$, then $\mathcal{C}_3$ and $\mathcal{C}_5$ are not monomially equivalent.
	\end{prop}
	\begin{proof}
		Suppose that $\mathcal{C}_3$ and $\mathcal{C}_5$ are  monomially equivalent. Let $\mathcal{S}_{h}$ and $\mathcal{S}_{1}$ be the  $(q+1)$-arcs corresponding to the first $q+1$ columns of  generator matrices of $\mathcal{C}_3$ and $\mathcal{C}_5$, respectively. By analogy to Lemma    \ref{equivalent}, there exists a  map $\theta\in\mathrm{PGL}(4,q)$ such that  $\mathcal{S}^{\theta}_{h}\cup\{\mathbf{g}^{\theta}_{q+2},\mathbf{g}^{\theta}_{q+3}\}=\mathcal{S}_{1}\cup\{\mathbf{g}_{q+2},\mathbf{g}_{q+3}\}$. This implies that $|\mathcal{S}^{\theta}_{h}\cap \mathcal{S}_{1}|\geq q-3>q-\sqrt{q}/2+9/4$.  By Lemma \ref{unique2}, it follows  that $\mathcal{S}^{\theta}_{h}=\mathcal{S}_{1}$,  i.e., $\mathcal{S}_{h}$ is projectively equivalent to  $\mathcal{S}_{1}$, which is a contradiction. This completes the proof.
	\end{proof}
	
	\section{Conclusion}\label{s5}
	
	In this paper, we  investigated the monomial equivalence and construction of NMDS codes with dimensions 3 and 4 using group theory and geometric methods. In dimension 3, we proposed a criterion for the monomial equivalence of NMDS codes derived from hyperovals in $\mathrm{PG}(2,q)$. By analyzing the actions of the homography stabilizers of distinct hyperovals on the points not in hyperovals, we determined the equivalence classes of NMDS codes derived from augmenting hyperovals with a single point. 
	In dimension 4,  we also determined  the equivalence classes of NMDS codes derived from arcs in $\mathrm{PG}(3,q)$  with additional points, which included  new constructions of NMDS codes of length $q+2$. By analyzing the orbits of points and planes  of $\mathrm{PG}(3,q)$ under the stabilizers of arcs,  we provided a unified geometric perspective for prior constructions of NMDS codes (e.g., \cite{Ding2024,Xu2023}) and simplified the proofs through geometric invariants instead of algebraic computations. Finally, while the stabilizer-based approach is highly effective for codes derived from   arcs or hyperovals, its applicability to longer, less structured NMDS codes remains a challenging and interesting open problem.
	\section{Acknowledgements}
	
This work was supported by the National Natural Science Foundation of China under Grant No. 12501467 and the National Key Research and Development Program of China under Grant No. 2022YFA1004900.
	\vspace*{10pt}

	\begin{center}
		\scriptsize
		\setlength{\bibsep}{0.5ex}  
		\linespread{0.5}

\begin{thebibliography}{10}
			
			\bibitem{Ball2012}
			S. Ball,
			\newblock	 On large subsets of a finite vector space in which every subset of basis size is a basis,
			\newblock {\em  J. Eur.	Math. Soc.}  14 (2012) 733--748.
			
			\bibitem{Ball}
			S.	Ball,  J. D. Beule,
			\newblock On sets of vectors of a finite vector space in which every subset of basis size is a basis II,
			\newblock {\em Des. Codes Cryptogr.} 65(1) (2012) 5–14.
			
			\bibitem{Ball2015}
			S. Ball,
			\newblock  {\it Finite Geometry and Combinatorial Applications}, Cambridge University Press, Cambridge, 2015.
			
			\bibitem{Bartoli2020}
			D. Bartoli, A. A. Davydov, S. Marcugini, F. Pambianco, 
			\newblock  On planes through points off the twisted cubic in $\mathrm{PG}(3,q)$ and multiple covering codes,
			\newblock {\em Finite Fields Appl.} 67 (2020) 101710.
			
			\bibitem{Bayens2007} 	
			L. Bayens, W. Cherowitzo, T. Penttila, 
			\newblock Groups of hyperovals in Desarguesian planes, 
			\newblock {\it Innov. Incidence Geom.} 6/7 (2007/08) 37--51.
			
			\bibitem{Bose1947}
			R. C. Bose, 
			\newblock 	Mathematical theory of the symmetrical factorial design,
			\newblock {\em Sankhya} 8 (1947) 107--166.
			
			\bibitem{Bosma}
			W. Bosma, J. Cannon, C. Fieker, A. Steel,
			\newblock {\it Handbook of Magma Functions}, 2017.
			
			\bibitem{Ceria}
			M. Ceria, A. Cossidente, G. Marino, F. Pavese, 
			\newblock	On near-MDS codes and caps,
			\newblock {\it Des. Codes Cryptogr.} 91(2023) 1095--1110.
			
			
			
			\bibitem{Ceria2023} 	
			M. Ceria, F. Pavese, 
			\newblock On the geometry of a $(q+1)$-arc of $\mathrm{PG}(3,q)$, $q$ even,
			\newblock {\it Discrete Math.} 346(12) (2023)  113594.
			
			\bibitem{Cherowitzo1998} 	
			W. E. Cherowitzo, 
			\newblock $\alpha$-flocks and hyperovals, 
			\newblock {\it	Geom. Dedicata} 72 (1998)  221--246.
			
			\bibitem{Cherowitzo2003} 	
			W.E. Cherowitzo, C.M. O'Keefe, T. Penttila, 
			\newblock A unified construction of finite geometries associated with $q$-clans in characteristic two,
			\newblock {\it Adv. Geom.} 3 (2003) 1--21.
			
			\bibitem{Ding2024} 	
			Y. Ding, Y. Li, S. X. Zhu. 
			\newblock Four new families of NMDS codes with dimension 4 and their applications
			\newblock {\it Finite Fields Appl.} 99 (2024) 99  102495.
			
			\bibitem{Dodunekov1995}
			S. Dodunekov, I. Landgev, 
			\newblock On near-MDS codes, 
			\newblock {\it J. Geom.} 54 (1995) 30--43.
			
			\bibitem{Fan2024}
			C. L. Fan, A. Wang, L. Xu, 
			\newblock	New classes of NMDS codes with dimension $3$, 
			\newblock {\it Des. Codes Cryptogr.} 92(2024) 397--418.
			
			\bibitem{Glynn1983}
			D. G. Glynn, 
			\newblock Two new sequences of ovals in finite Desarguesian planes of even order, 
			\newblock  {\it Combinatorial Mathematics X} (ed. L. R. A. Casse), Lecture Notes in Mathematics, vol. 1036, Springer, Berlin, 1983, pp. 217--229.
			
			\bibitem{Heng} 
			Z. L. Heng, X. R. Li, 
			\newblock Near MDS codes with dimension 4 and their application in locally recoverable codes, 
			in: S. Mesnager, Z. Zhou (Eds.), {\it Arithmetic of Finite Fields}, WAIFI 2022, in: Lecture Notesin Computer Science, vol. 13638, pp.142--158, Springer, Cham, 2023.
			
			\bibitem{Hirschfeld1979} 
			J. W. P. Hirschfeld, 
			\newblock {\it Projective Geometries over Finite Fields}, Oxford Univ. Press, Oxford, 1979.
			
			\bibitem{Hirschfeld1985} 
			J. W. P. Hirschfeld, 
			\newblock {\it Finite Projective Spaces of Three Dimensions}, Oxford Mathematical Monographs, Oxford Science Publications, The Clarendon Press, 
			Oxford University Press, New York, 1985.
			
			\bibitem{Hirschfeld2001} 
			J. W. P. Hirschfeld, L. Storme, 
			\newblock The packing problem in statistics, coding theory and finite projective spaces: update 2001, 
			\newblock {\em in: Finite Geometries, Developments in Mathematics}, Kluwer, Boston, 2001,pp. 201--246.
			
			\bibitem{Landjev2015} 
			I. Landjev, A. Rousseva, 
			\newblock The main conjecture for near-MDS codes, 
			in: {\em Proc. 9th INt. WorkshopCoding Cryptogr.} (WCC), 2015.
			
			\bibitem{Li20}
			X. R. Li, Z. L. Heng, 
			\newblock A construction of optimal locally recoverable codes, 
			\newblock {\em  Cryptogr. Commun.} 15 (2023) 553--563.
			
			\bibitem{Li2023}
			X. R. Li, Z. L. Heng,
			\newblock Constructions of near MDS codes which are optimal locally recoverable codes, 
			\newblock {\em Finite Fields Appl.} 88 (2023) 102184.
			
			\bibitem{Lunelli1958}
			L. Lunelli, M. Sce, 
			\newblock $k$-archi completi nei piani proiettivi desarguesiani di
			rango 8 e 16, 
			\newblock {\em Centro di Calcoli Numerici}, Politecnico di Milano, 1958.
			
			\bibitem{Maschietti1998}
			A. Maschietti, Difference set and hyperovals, 
			\newblock {\em Des. Codes Cryptogr.} 14 (1998) 89--98
			
			\bibitem{O'Keefe1992}
			C. M. O'Keefe, T. Penttila, 
			\newblock A new hyperoval in $\mathrm{PG}(2,32)$, 
			\newblock {\em J. Geom.} 44 (1992) 117--139.
			
			\bibitem{O'Keefe1994} 
			C. M. O'Keefe, T. Penttila, 
			\newblock Symmetries of arcs, 
			\newblock {\em J. Combin. Theory Ser. A}  66 (1994) 53--67. 
			
			\bibitem{O'Keefe1996} 
			C. M. O'Keefe, J. A. Thas, 
			\newblock Collineations of Subiaco and Cherowitzo hyperovals, 
			\newblock {\em Bull. Belg. Math. Soc. Simon Stevin}  3 (1996) 177--192.
			
			
			
			
			
			\bibitem{Payne1985} 
			S. E. Payne,
			\newblock A new infinite family of generalized quadrangles, 
			\newblock {\em Congr. Numer.} 49 (1985) 115--128.
			
			
			\bibitem{Payne} 
			S. E. Payne,
			\newblock  {\it Topics in finite geometry: ovals, ovoids and generalized quadrangles}, http://math.ucdenver.edu/~spayne/classnotes/topics.pdf, edition of 16 May 2007.
			
			\bibitem{Payne1978} 
			S. E. Payne, J. E. Conklin, 
			\newblock An unusual generalized quadrangle of order sixteen,
			\newblock {\em J. Combin. Theory Ser. A} 24 (1978)  50--74.
			
			\bibitem{Payne1995} 
			S. E. Payne, T. Penttila, I. Pinneri, 
			Isomorphisms between Subiaco $q$-clan	geometries, 
			\newblock {\em Bull. Belgian Math. Soc.}   Simon Stevin, 2 (1995)  197--222.
			
			\bibitem{Payne2005}
			S. Payne, J. A. Thas, 
			\newblock The stabilizer of the Adelaide oval, 
			\newblock {\em	Discrete Math.} 294 (2005) 16--173.
			
			\bibitem{Penttila1994}
			T. Penttila, I. Pinneri, 
			\newblock  Irregular hyperovals in $\mathrm{PG}(2,64)$, 
			\newblock {\em	J. Geom.} 51 (1994) 89--100.
			
			\bibitem{Segre1955} 
			B. Segre,
			\newblock  Curve razionali normali $ek$-archi negli spazi finiti, 
			\newblock {\em	Ann. Mat. Pura Appl.} 39(1) (1955) 357--379. 
			
			\bibitem{Segre1957}
			B. Segre, 
			\newblock Sui $k$-archi nei piani finite de caratteristica due, 
			\newblock {\em Rev. Math. Pures Appl.} 2 (1957)  289--300.
			
			\bibitem{Segre1962}
			B. Segre,
			Ovali e curve $\sigma$ nei piani di Galois di caratteristica due,
			\newblock {\em Atti dell’ Accad. Naz. Lincei Rend.} 32(8) 1962 785--790.
			
			\bibitem{Singleton1964}
			R. Singleton,
			\newblock Maximum distance q-nary codes,
			\newblock {\em IEEE Trans. Inf. Theory} 10 (2) (1964) 116--118.
			
			\bibitem{Storme1993} 
			L. Storme, J. A. Thas, M.D.S. codes and arcs in $\mathrm{PG}(n,q)$ with $q$ even: an improvement of the bounds of Bruen, Thas, and Blokhuis, 
			\newblock {\em J. Comb. Theory, Ser. A} 62 (1993) 139--154.
			
			\bibitem{Thas1988}
			J. A. Thas, S. E. Payne, H. Gevaert, 
			\newblock	A family of ovals with few collineations,
			\newblock {\em	European J. Combin.} 9 (1988)  353--362.
			
			\bibitem{Wan2003} 
			Z. X. Wan, 
			\newblock {\it Lectures on Finite Fields and Galois Rings}, World Scientific Publishing Company, 2003.
			
			
			
			
			\bibitem{Wang2021}
			Q. Y. Wang, Z. L. Heng, 
			\newblock  Near MDS codes from oval polynomials,
			\newblock {\em Discrete Math.} 344(4) (2021) 112277.
			
			\bibitem{Xu2023}
			L. Xu, C. L. Fan, 
			\newblock Near MDS codes of non-elliptic-curve type from Reed-Solomon codes, \newblock {\em Discrete Math.} 346(9) (2023) 113490.
			
			\bibitem{Xu2024}
			L. Xu, C. L. Fan, D. C. Han, 
			\newblock Near-MDS codes from maximal arcs in $\mathrm{PG}(2,q)$, 
			\newblock {\em Finite Fields Appl. } 93 (2024) 102338.
			
			
			
		\end{thebibliography}
		\bibliographystyle{plain}

	\end{center}
	
\end{document}